\documentclass[12pt, draftclsnofoot, onecolumn]{IEEEtran}

\usepackage{cite}

\ifCLASSINFOpdf
   \usepackage[pdftex]{graphicx}
\else
\fi

\usepackage[cmex10]{amsmath}
\usepackage{amssymb}
\usepackage{algorithmic}
\usepackage{varwidth}
\usepackage{subfigure}
\usepackage{bm}
\usepackage{algorithm}
\usepackage{amssymb}
\usepackage{amsmath}
\usepackage{mathtools}
\usepackage{amsthm}
\usepackage{cite}
\usepackage{multicol}
\usepackage{graphicx}
\usepackage{xcolor}
\usepackage{soul}
\newtheorem{theorem}{Theorem}
\newtheorem{corollary}{Corollary}[theorem]

\begin{document}
\title{Dictionary Learning Based Sparse Channel Representation and Estimation for FDD Massive MIMO Systems}

\author{Yacong~Ding,~\IEEEmembership{Student Member,~IEEE,}
        Bhaskar~D.~Rao,~\IEEEmembership{Fellow,~IEEE}
\thanks{Y. Ding and B. D. Rao are with Department of Electrical and Computer Engineering, University of California, San Diego, La Jolla, CA, 92093-0407 USA (e-mail: \{yad003,brao\}@ucsd.edu).}
}

\maketitle

\begin{abstract}
This paper addresses the problem of uplink and downlink channel estimation in FDD Massive MIMO systems. By utilizing sparse recovery and compressive sensing algorithms, we are able to improve the accuracy of the uplink/downlink channel estimation and reduce the number of uplink/downlink pilot symbols. Such successful channel estimation builds upon the assumption that the channel can be sparsely represented under some basis/dictionary. Previous works model the channel using some predefined basis/dictionary, while in this work, we present a dictionary learning based channel model such that a dictionary is learned from comprehensively collected channel measurements. The learned dictionary adapts specifically to the cell characteristics and promotes a more efficient and robust channel representation, which in turn improves the performance of the channel estimation. Furthermore, we extend the dictionary learning based channel model into a joint uplink/downlink learning framework by observing the reciprocity of the AOA/AOD between the uplink/downlink transmission, and propose a joint channel estimation algorithm that combines the uplink and downlink received training signals to obtain a more accurate channel estimate. In other words, the downlink training overhead, which is a bottleneck in FDD Massive MIMO system, can be reduced by utilizing the information from simpler uplink training.
\end{abstract}


\IEEEpeerreviewmaketitle

\section{Introduction}
\IEEEPARstart{M}{assive} Multiple-Input Multiple-Output (MIMO) systems have been proposed for the next generation of communication systems. 
By deploying a large antenna array at the base station (BS), both receive combining and transmit beamforming can be performed with narrow beams, thereby eliminating multiuser interference and  increasing the cell throughput.
For effective uplink (UL) combining and downlink (DL) precoding, it is essential to have accurate knowledge of the channel state information (CSI) at BS. The common assumption in Massive MIMO is that each user equipment (UE) only has a small number of antennas, therefore it is relatively easy to have the uplink CSI since the uplink training overhead is only proportional to the number of users \cite{marzetta2006much}. In a time-division duplexing (TDD) system, downlink CSI can also be easily obtained by exploiting the uplink/downlink channel reciprocity.
On the other hand, channel reciprocity is no longer valid in a frequency-division duplexing (FDD) system because the uplink and downlink transmission are operated at different frequencies. In order to have downlink CSI, the BS has to perform downlink training. Subsequently, the user needs to estimate, quantize and feedback the channel state information. When conventional channel estimation and feedback schemes are used, the downlink training and feedback overhead are proportional to the number of antennas at the base station. The large antenna array in the Massive MIMO system makes such training impractical due to the high overhead and infeasible when the coherence time of the channel is limited. However, since FDD system is generally considered to be more effective for systems with symmetric traffic and delay-sensitive applications, most cellular systems today employ FDD \cite{adhikary2013joint, rao2014distributed}.

To alleviate the overhead of downlink channel training and feedback in a FDD Massive MIMO system, one option is to explore possible underlying channel structure whereby the high dimensional channel vector has a low dimensional representation \cite{bajwa2010compressed, adhikary2013joint, rao2014distributed}. Motivated by the framework of \emph{Compressive Sensing} (CS), if the desired signal (channel response) can be sparsely represented in some basis or dictionary, then it can be robustly recovered with the number of measurements (downlink pilot symbols) only proportional to the number of nonzero entries in the representation
\cite{candes2008introduction}. This indicates that when such basis or dictionary does exist and leads to a very sparse representation, we are able to greatly reduce the downlink training overhead. Fortunately, the limited scattering environment implies the low dimensionality of the channel, and the large antenna array provides finer angular resolution to resolve the limited scattering and represent channel sparsely \cite{sayeed2002deconstructing, tse2005fundamentals}. Many previous works have proposed efficient downlink channel estimation and feedback algorithms based on this sparse assumption \cite{bajwa2010compressed, adhikary2013joint, rao2014distributed, gao2015spatially, shen2016compressed, zhou2016sparse,han2017compressed}.

In this paper, besides the downlink channel estimation, we also utilize the sparse properties of the channel for the uplink channel estimation. What is more, a new channel modeling framework based on \emph{learning} techniques is developed, 
and is extended into a joint uplink/downlink channel representation by observing the reciprocity between the uplink and downlink transmission. 
In the following, we review the previous works and summarize the contributions. Preliminary versions of this work have appeared in \cite{ding2015compressed, ding2015channel}:

\textit{1)} We formulate the uplink channel estimation explicitly into a sparse recovery problem. Although the compressive sensing formulation has been applied widely in the downlink channel estimation, utilizing sparse property for the uplink has only received limited attention \cite{nguyen2013compressive,ding2018bayesian}. We show that with both appropriate pilots design and non-overlapping (or limited overlapping in practice) sparse supports of users, good estimation accuracy can be achieved even with pilot symbols \textit{less} than the number of users, which is the underdetermined case for conventional least square channel estimation. 

\textit{2)} We propose a \emph{dictionary learning based channel model} (DLCM), where a learned \emph{overcomplete} dictionary is used to represent the channel in some \emph{specific cell}. To learn the dictionary, a large number of channel measurements need to be collected from different locations in a specific cell at the cell deployment stage, and used as the training samples for the dictionary learning algorithm. The learned dictionary is able to adapt to the cell characteristics as well as ensure a sparse representation of the channel. Since  no structural constraints are placed on the dictionary, the approach is applicable to an arbitrary array geometry and does \textit{not} require accurate array calibration. We demonstrate the improved channel estimation performance when applying the learned dictionary, compared to existing works which utilize some predefined basis. In \cite{schniter2014channel}, an aperture shaping scheme has been proposed that promotes sparse representation in the virtual channel model. Notice that the dictionary learning concept itself has been widely investigated in previous works \cite{engan1999method, kreutz2003dictionary, aharon2006svd}, with many applications such as image denoising and feature extraction. But to the best of our knowledge, our work is the first to utilize the dictionary learning framework to model the Massive MIMO channel.

\textit{3)} We develop a general framework of \textit{joint} uplink/downlink dictionary learning based channel model (JDLCM) and channel estimation by observing the reciprocity resides in the uplink and downlink channels. Although in FDD systems the uplink and downlink are operated in different frequency band, the propagation environment is the same for the uplink and downlink transmission when the duplex distance is not large \cite{paulraj1997space, hugl2002spatial}. This motivates a joint sparse representation of uplink and downlink channels, and enables the use of information from the more easily obtained uplink training to help downlink channel estimation. In FDD systems, leveraging uplink channel information for the downlink use has been proposed, for example using uplink signals to compute direction of arrival (DOA) and construct downlink channel response \cite{zetterberg1995spectrum}
or utilizing uplink channel covariance matrix to estimate downlink channel covariance matrix \cite{hugl1999downlink}.
To the best of our knowledge, our work is the first to explore the jointly sparse representation as an abstract model for the uplink and downlink channel reciprocity, and develop joint channel estimation algorithms to improve the channel estimation performance.

Notations used in this paper are as follows. Upper (lower) bold face letters are used throughout to denote matrices (column vectors). \((\cdot)^T\), \((\cdot)^H\) \((\cdot)^{\dagger}\) denotes the transpose, Hermitian transpose, and the Moore-Penrose pseudo-inverse. \(\bm{A}_{i\cdot}\) and \(\bm{A}_{\cdot j}\) represents the \(i\)-th row and \(j\)-th column of \(\bm{A}\), and for a set \(\mathcal{S}\) we denote \(\bm{A}_{\mathcal{S}}\) to be the submatrix of \(\bm A\) that contains columns indexed by elements of \(\mathcal{S}\).
For a vector \(\bm x\), \(\text{diag}(\bm x)\) is a diagonal matrix with entries of \(\bm x\) along its diagonal. \(\|\bm x\|_1,\|\bm x\|_2\) denotes the \(\ell_1\) and \(\ell_2\) norm. \(\|\bm x\|_0\) represents the number of nonzero entries in \(\bm x\) and is referred to as the \(\ell_0\) norm. \(\text{supp}(\bm x)\) denotes the set of indices such that the corresponding entries of \(\bm x\) are nonzero.

\section{Channel Estimation Based on Sparse Channel Model}\label{section: system model}
\subsection{Physical Channel Model}
We consider a single cell Massive MIMO system operated in FDD mode. The BS is equipped with \(N\) antennas and each UE has a single antenna. Assume a narrowband block fading channel, we adopt a simplified spatial channel model which captures the physical propagation structure of either the uplink or the downlink transmission as
\begin{equation}\label{equ: SCM_downlink}
\begin{aligned}
\bm h = \sum_{i=1}^{N_{c}}\sum_{l=1}^{N_{s}}\alpha_{il}\bm a(\Omega_{il})
\end{aligned}
\end{equation}
where \(N_{c}\) is the number of scattering clusters, each of which contains \(N_{s}\) propagation subpaths. \(\alpha_{il}\) is the complex gains of the \(l\)-th subpath in the \(i\)-th scattering cluster for the uplink or the downlink. For 2D channel model \cite{sayeed2002deconstructing, tse2005fundamentals, 3gpp.25.996}, \(\Omega_{il} = \{\theta_{il}\}\) denotes the angle of arrival (AOA) for the uplink transmission or the angle of departure (AOD) for the downlink. \(\bm a(\Omega_{il})\) is the array response vectors, and for a uniform linear array (ULA) 
\begin{equation}\label{equ: array response}
\begin{aligned}
&\bm a(\theta) = [1, e^{j2\pi\frac{d}{\lambda}\text{sin}(\theta)}, \ldots , e^{j2\pi\frac{d}{\lambda}\text{sin}(\theta)\cdot(N-1)}]^T \\
\end{aligned}
\end{equation}
where \(d\) is the antenna spacing and \(\lambda\) is the wavelength of uplink or downlink propagation. 
For 3D channel model \cite{ying2014kronecker,3gpp.38.901}, \(\Omega_{il} = \{\theta_{il}, \phi_{il}\}\), where \(\theta_{il}, \phi_{il}\) denotes the zenith angle of arrival (ZOA) and azimuth angle of arrival (AOA) for the uplink, and zenith angle of departure (ZOD) and azimuth angle of departure (AOD) for the downlink. For a uniform rectangular array (URA) with \(N_1\) vertical antennas spaced by \(d_1\) and \(N_2\) horizontal antennas with \(d_2\) spacing, \(N_1N_2= N\), the array response vectors is given as \cite{ying2014kronecker}
\begin{equation}\label{equ: array response URA}
\begin{aligned}
&\bm a(\theta, \phi) = \bm q(v) \otimes \bm p(u) \\
\end{aligned}
\end{equation}
where \(\bm p(u) = [1, e^{ju}, \ldots , e^{j(N_1-1)u}]^T\), \(\bm q(v) = [1, e^{jv}, \ldots , e^{j(N_2-1)v}]^T\), \(u = 2\pi d_1 \text{cos}(\theta)/\lambda\), and \(v = 2\pi d_2 \text{sin}(\theta) \text{cos}(\phi)/\lambda\).

In order to model the scattering clusters, we consider the principles of \emph{Geometry-Based Stochastic Channel Model} (GSCM) \cite{molisch2003geometry}, as illustrated in Fig. \ref{fig:cell_downlink}. For a specific cell, the locations of  the dominant scattering clusters are determined by cell specific attributes such as the buildings, and are common to all the users irrespective of user position.
We assume such scattering clusters are far away from the base station, so the subpaths associated with a specific scattering cluster will be concentrated in a small range, 
i.e., having a small angular spread (AS). While modeling the scattering effects which are user-location dependent, for example the ground reflection close to the user, or some moving physical scatterers near the user,  we assume the UE is far away from the base station, so subpaths associated with the user-location dependent scattering cluster also have small angular spread.
Since the BS is far away and is commonly assumed to be mounted at a height, the number of scattering clusters that contribute to the channel responses is limited, i.e., \(N_c\) is small. Because the number of scattering clusters is limited and each of them spans a small AS, there are only limited dimensions being occupied when viewed from the angular domain. Furthermore, the large antenna array at the BS leads to narrower beamwidth, resulting in smaller leakage effect of some scattering cluster to the other angular bins.
Due to the limited scattering effect and the large antenna array, it is reasonable to assume a low dimensional representation for the large Massive MIMO channel \cite{bajwa2010compressed, adhikary2013joint, rao2014distributed, gao2015spatially, shen2016compressed,han2017compressed}.

\begin{figure}[!t]
  \centering
  \includegraphics[width=0.47\textwidth]{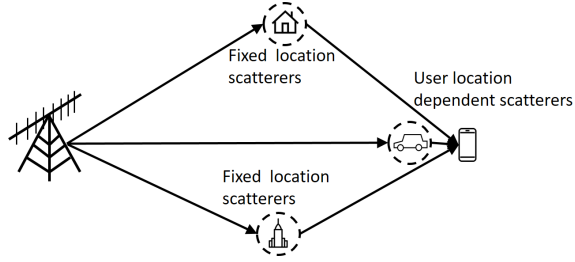}
\caption{Illustration of signal propagation in a typical cell}
\label{fig:cell_downlink}
\end{figure}

\subsection{Compressive Sensing Based Downlink Channel Estimation}
For the downlink channel estimation in FDD system, the BS transmits training pilots. The UE estimates the channel and feed back the channel state information to the BS. The received signal \(\bm y^d\) at the UE is given as
\begin{equation}
\bm y^d = \bm A \bm h^d +\bm w^d
\end{equation}
where \(\bm h^d\in\mathbb{C}^{N\times1}\) denotes the downlink channel response, \(\bm w^d\in\mathbb{C}^{T\times 1}\) is the received noise vector such that \(\bm w^d\sim\mathcal{CN}(0,\textbf{I})\). \(\bm A\in\mathbb{C}^{T^d\times N}\) is the downlink pilots transmitted during the training period of \(T^d\) symbols, where \(\|\bm A\|_F^2= \rho^d T^d\) such that \(\rho^d\) measures the training SNR.
Using conventional channel estimation technique such as \emph{Least Square} (LS) channel estimation,
the estimated channel is given by
\begin{equation}\label{equ:LS channel estiamtion}
\hat{\bm h}_{LS}^d=\bm A^{\dagger}\bm y^d
\end{equation}
where \(\bm A^{\dagger}\) is the Moore-Penrose pseudoinverse. Robust recovery of \(\bm h^d\) by LS channel estimation requires \(T^d \ge N\), which means the training period has to be larger than the number of antennas. In a Massive MIMO system \(N\) is very large making this infeasible. Moreover, the UE needs to feed back channel information to the BS, which also requires feedback resources proportional to channel dimension \(N\). The finite channel coherence time further exacerbates the situation.

In order to robustly estimate downlink channel with limited training overhead, compressive sensing based channel estimation has been proposed in previous works \cite{bajwa2010compressed, rao2014distributed, gao2015spatially, shen2016compressed, zhou2016sparse,han2017compressed}, and we briefly review the main steps in the following. In the compressive sensing framework, methods to measure a high dimensional signal (\(\bm h^d\)) have been proposed with much smaller measurements (\(T^d<N\)), provided the original signal can be sparsely represented in some sparsifying basis \cite{candes2008introduction}. Assume there exists a \emph{sparsifying matrix} \(\bm D^d\in\mathbb{C}^{N\times M}\) (\(M\ge N\)) such that \(\bm h^d=\bm D^d\bm \beta^d\), where the  representation vector \(\bm \beta^d\in\mathbb{C}^{M\times 1}\) is sparse, i.e., \(s = \|\bm \beta^d \|_0\ll N\).
Then the downlink channel estimation can be written as
\begin{equation}\label{equ: downlink compressed training}
\bm y^d = \bm A \bm h^d +\bm w^d = \bm A \bm D^d \bm \beta^d + \bm w^d.
\end{equation}
Given \(\bm y^d, \bm A\) and \(\bm D^d\), if we are able to solve for \(\bm \beta^d\), then the channel estimate is obtained as \(\hat{\bm h}^d=\bm D^d \bm \beta^d\). However, (\ref{equ: downlink compressed training}) is an underdetermined system if we plan to use a small number of training samples \(T^d<N\). The system will in general have an infinite number of solutions for \(\bm \beta^d\) and the sparsity assumption provides a mechanism to regularize the problem. Consider the minimum sparsity assumption that \(s\ll N\) and assuming \(\|\bm w^d\|_2\le \epsilon\), then the problem reduces to
\begin{equation}\label{equ: l0 minimization}
\begin{aligned}
&\hat{\bm \beta}^d=\arg\min_{\bm \beta^d} \|\bm \beta^d\|_0 \quad \text{subject to} \quad \|\bm y^d -\bm A\bm D^d\bm \beta^d\|_2\le \epsilon\\
\end{aligned}
\end{equation}
and \(\hat{\bm h}_{CS}^d=\bm D^d\hat{\bm \beta}^d\).
Notice that the optimization formula in (\ref{equ: l0 minimization}) is non-convex,  and a number of suboptimal but effective algorithms have been proposed to solve the problem \cite{choi2017compressed}. One of the most widely used framework is to relax the \(\ell_0\) norm \(\|\bm \beta^d\|_0\) to the \(\ell_1\) norm \(\|\bm \beta^d\|_1\). It has been shown that under certain conditions on \(\bm A\bm D^d\),  based on the $\ell_1$ norm criteria a solution of \(\bm \beta^d\) with bounded error can be obtained with \(T^d\ge c\cdot s\text{log}(N/s)\), where \(c\) is some constant\cite{candes2008introduction}. Instead of using a training period proportional to the channel dimension \(N\), we can compute good channel estimate with training period proportional to sparsity level \(s\), which is assumed to be much less than \(N\). This  makes downlink channel estimation feasible in a limited training period.

The CS based downlink channel estimation in (\ref{equ: l0 minimization}) is for single antenna at the UE, and we show in the following how to extend it to scenario where UE has multiple antennas. Assume \(N_T\) antennas at BS and \(N_R\) antennas at UE, then the channel \(\bm H^d = [\bm h_1^d, \ldots, \bm h_{N_R}^d]\in \mathbb{C}^{N_T \times N_R}\) where the column \(\bm h_k^d\) denotes the channel from \(N_T\) BS antennas to the \(k\)-th UE antenna. Since the antenna aperture at UE is much smaller than the distance between the antenna and the scattering clusters in the environment, the scattering clusters that affect the signal transmission are the same for all antenna elements at the UE side. With the sparse representation \(\bm h_k^d = \bm D^d\bm \beta_k^d,\forall k\), it implies the support of \(\bm \beta_k^d\) is the same for all \(N_R\) antennas, i.e., \(\text{supp}(\bm \beta_1^d)=\ldots=\text{supp}(\bm \beta_{N_R}^d)\). Denoting \(\bm B^d = [\bm \beta_1^d,\ldots,\bm \beta_{N_R}^d]\), then \(\bm H^d = \bm D^d \bm B^d\) and the matrix \(\bm B^d\) is row sparse. 
Similar observation can be made when applying the virtual channel model \(\bm H^d = \bm A_T \tilde{\bm H}^d \bm A_R^H\) \cite{bajwa2010compressed,tse2005fundamentals}, where \(\bm A_T\in \mathbb{C}^{N_T \times N_T}\) and \(\bm A_R\in \mathbb{C}^{N_R \times N_R}\) are orthogonal DFT matrices, \(\bm \tilde{\bm H}^d\) contains the virtual channel coefficients and is assumed to be sparse. Assume the \(i\)-th row \(\tilde{\bm H}_{i\cdot}^d = \mathbf{0}_{1\times N_R}\), then the whole \(i\)-th row of the combined matrix \(\tilde{\bm H}^d \bm A_R^H\) (act similarly as \(\bm B^d\)) is zero, implying the row sparsity of the matrix \(\tilde{\bm H}^d \bm A_R^H\). The downlink training can be written as
\begin{equation}\label{equ: downlink compressed training MMV}
\bm Y^d = \bm A \bm H^d +\bm W^d = \bm A \bm D^d \bm B^d + \bm W^d
\end{equation}
where \(\bm Y^d\in \mathbb{C}^{T^d\times N_R}\). With respect to the row sparsity of \(\bm B^d\), we cast the channel estimation into solving a multiple measurement vector (MMV) problem \cite{cotter2005sparse,tropp2006algorithms,wipf2007empirical,zhang2011sparse} such as
\begin{equation}\label{equ: MMV minimization}
\begin{aligned}
&\hat{\bm B}^d=\arg\min_{\bm B^d} \|\bm B^d\|_{1,2} \enspace \text{subject to} \enspace \|\bm Y^d -\bm A\bm D^d\bm B^d\|_F\le \epsilon,\\
\end{aligned}
\end{equation}
where \(\|\bm B^d\|_{1,2}=\sum_{i=1}^M \|\bm B^d_{i\cdot}\|_2\), i.e., the summation of the \(\ell_2\) norm of each row in \(\bm B^d\). The estimated channel is given by \(\hat{\bm H}^d_{CS}=\bm D^d \hat{\bm B}^d\). For the sparse recovery, it has been shown that utilizing the row sparse property in the MMV formulation can achieve better recovery performance compared to the single measurement vector (SMV) \cite{cotter2005sparse,tropp2006algorithms,wipf2007empirical,zhang2011sparse}. In Section \ref{section: simulation}, we show the improved channel estimation performance when having multiple antennas at the UE. 

\subsection{Sparse Recovery Based Uplink Channel Estimation}
\label{subsection: uplink channel estimation}
In contrast to the downlink channel estimation, uplink channel estimation is relatively easy in a Massive MIMO system.
With the same assumption of \(N\) antennas at the BS and a single antenna at the UE, for \(K\) UEs the uplink training can be written as
\begin{equation}\label{equ: uplink training}
\begin{aligned}
\bm Y^u
&= \sum_{k=1}^K \bm h^u_k \sqrt{\rho^u_k T^u} \bm s_k^T + \bm W^u = \bm H^u \bm C\bm S +\bm W^u
\end{aligned}
\end{equation}
where \(\bm H^u = [\bm h^u_1,\ldots, \bm h^u_K]\in\mathbb{C}^{N\times K}\) is the uplink channel for \(K\) UEs, \(\bm Y^u\in\mathbb{C}^{N\times T^u}\) denotes the received signal at the base station and \(\bm W^u\in\mathbb{C}^{N\times T^u}\) is the received noise whose elements are assumed to be i.i.d Gaussian with zero mean and unit variance.
\(\bm S=[\bm s_1,\ldots,\bm s_K]^T\in\mathbb{C}^{K\times T^u}\) denotes the uplink pilots during training period \(T^u\), where \(\|\bm s_k\|_2^2=1\). \(\rho^u_k\) denotes the uplink training SNR for the \(k\)-th UE, which incorporates the transmit power, path loss and shadow fading, and is assumed to change slowly and known \textit{a priori}. \(\bm C = \text{diag}(\sqrt{\rho^u_1 T^u},\ldots, \sqrt{\rho^u_K T^u})\).
Using LS channel estimation, we have
\begin{equation}\label{equ: LS uplink}
\hat{\bm H}^u_{LS} = \bm Y^u(\bm C\bm S)^{\dagger}.
\end{equation}
For the robust estimation, we only require \(T^u\ge K\), i.e., the number of pilots to be greater than the number of users. In Massive MIMO systems, it is common to assume the number of users is much smaller than the number of antennas. Comparing to \(T^d\ge N\) for the downlink estimation,  the uplink channel estimation task is simpler. Moreover, the uplink channel is estimated at the BS, incurring no feedback overhead. 
In the following,  we show that the uplink channel can  be accurately estimated even when \(T^u <K\) by casting the channel estimation problem into a sparse recovery problem.

Assume each UE's channel \(\bm h^u_k=\bm D^u \bm \beta^u_k,\forall k\), where \(\bm D^u\in\mathbb{C}^{N\times M}\) is the \emph{sparsifying matrix} and \(\|\bm \beta^u_k \|_0\ll N\). Denoting \(\bm B^u = [\bm \beta^u_1,\ldots,\bm \beta^u_K]\), (\ref{equ: uplink training}) can be written as \(\bm Y^u= \bm H^u \bm C \bm S +\bm W^u = \bm D^u \bm B^u \bm C\bm S + \bm W^u\). Let \(\bm y^u=\textbf{vec}(\bm Y^u)\), we have
\begin{equation}\label{equ: uplink sparse reocvery}
\begin{aligned}
\bm y^u
&=(\bm S^T\otimes \bm D^u)\textbf{vec}(\bm B^u \bm C) + \textbf{vec}(\bm W^u)
=\bm E\bm b^u + \bm w^u\\
\end{aligned}
\end{equation}
where \(\bm E=\bm S^T\otimes\bm D^u\in\mathbb{C}^{NT^u\times MK}\) denotes the equivalent sparsifying matrix, \(\bm b^u =\textbf{vec}(\bm B^u \bm C) = \big[(\sqrt{\rho^u_1 T^u} \bm \beta_1^u)^T,\ldots,(\sqrt{\rho^u_K T^u}\bm \beta_K^u)^T \big]^T\) is the concatenated sparse coefficients. If \(\bm b^u\) is a sparse vector, i.e., \(\|\bm b^u\|_0 =\sum_{k=1}^K\|\bm \beta_k^u\|_0\ll NT^u\) then we can robustly estimate \(\bm b^u\) by many sparse recovery algorithms even when \(T^u<K\). Once \(\bm b^u\) is estimated, the uplink channel for the user \(k\) is given by \(\hat{\bm h}^u_k = \bm D^u \bm \beta^u_k\) since \(\rho^u_k\) is assumed to be known. Notice that \(T^u<K\) means the number of pilots is less than the number of users, which is underdetermined if using LS channel estimation. We denote the formulation in (\ref{equ: uplink sparse reocvery}) as the \emph{sparse recovery based channel estimation}, in contrast to the downlink compressive sensing based channel estimation (\ref{equ: downlink compressed training}) since there are no compressed measurements in the uplink.

In order to apply the sparse recovery algorithm, columns in \(\bm E\) are expected to be incoherent to each other, since two closely related columns may confuse any sparse recovery algorithm. Moreover, denoting \(\Lambda=\text{supp}(\bm b^u)\) where \(|\Lambda|=\|\bm b^u\|_0<NT^u\), given \(\Lambda\) \emph{a priori} the sparse recovery problem in (\ref{equ: uplink sparse reocvery}) reduces to \(\bm y^u = \bm E_{\Lambda}\bm b^u_{\Lambda} +\bm w^u\)
which can be solved by LS estimation. In this case, \(\bm E_{\Lambda}\) is required to be a well conditioned matrix for the robust LS estimation. To summary, we hope columns in \(\bm E\) and \(\bm E_{\Lambda}\) to be as uncorrelated to each other as possible. In the following, we show how to decrease the correlation of columns in \(\bm E\) and \(\bm E_{\Lambda}\) by designing uplink training pilots \(\bm S\) and performing uplink user scheduling.

To quantitatively characterize the correlation between columns in a matrix \(\bm X\),
we utilize the \emph{mutual coherence}
\footnote{Several other measures, e.g., null sparse property (NSP), restricted isometry property (RIP), etc., can provide better characterization of the geometry of a matrix. However those measures are difficult to evaluate explicitly \cite{foucart2013mathematical}.}
 \cite{foucart2013mathematical},
which is defined as the largest absolute and normalized inner product between different columns. Formally,
\begin{equation}\label{equ: mutual coherence}
\mu\{\bm X\} = \max_{i\neq j}\frac{|\bm X_{\cdot i}^H \bm X_{\cdot j}|}{\|\bm X_{\cdot i}\|\cdot \|\bm X_{\cdot j}\|}.
\end{equation}
The mutual coherence provides a measure of the worst similarity between the columns of \(\bm X\), which motivates us to minimize \(\mu\{\bm E\}\) and \(\mu\{\bm E_{\Lambda}\}\) to obtain a matrix with uncorrelated columns. Following this intuition, we first consider \(\mu\{\bm E\}\), which is described in the following theorem \cite{jokar2009sparse}
\footnote{Theorem \ref{theorem: muE} has been proved in \cite{jokar2009sparse}. Since we use slightly different steps in the proof, and the steps are also needed in the proof of Corollary \ref{corollary: muE_Lambda}, we include the proof steps for clarity.}
:
\begin{theorem}[\cite{jokar2009sparse}]
\label{theorem: muE}
Given \(\bm E = \bm S^T\otimes\bm D^u\) and the mutual coherence defined in (\ref{equ: mutual coherence}), \(\mu\{\bm E\}= \text{max}\{\mu\{\bm S^T\},\mu\{\bm D^u\}\}\).
\end{theorem}
\begin{proof}
To simplify the notation, denote \(\bm d_i=\bm D^u_{i}, \bm d_j =\bm D^u_{j}\), and \(\bm e_{li}=\bm E_{[(l-1)M+i]}=\bm s_l\otimes \bm d_i,\bm e_{kj}=\bm E_{[(k-1)M+j]}=\bm s_k\otimes \bm d_j,1\le i,j\le M, 1\le l,k\le K\).
Then we have \(\|\bm e_{li}\|_2^2=\bm e_{li}^H \bm e_{li}=\bm (\bm s_l\otimes \bm d_i)^H(\bm s_l\otimes \bm d_i)=(\bm s_l^H \bm s_l)\otimes(\bm d_i^H \bm d_i)=\|\bm s_l\|_2^2\|\bm d_i\|_2^2\).
So \(\|\bm e_{li}\|_2=\|\bm s_l\|_2\|\bm d_i\|_2\), and \(\|\bm e_{kj}\|_2=\|\bm s_k\|_2\|\bm d_j\|_2\). Similarly, \(|\bm e_{li}^H\bm e_{kj}|=|(\bm s_l\otimes \bm d_i)^H(\bm s_k\otimes \bm d_j)| = |(\bm s_l^H\bm s_k)\otimes(\bm d_i^H \bm d_j)| = |\bm s_l^H\bm s_k| |\bm d_i^H \bm d_j|\). According to (\ref{equ: mutual coherence}), the mutual coherence can be written as
\begin{equation}\label{equ: mutual e}
\begin{aligned}
\mu\{\bm E\}
&= \max_{\substack{(l-1)M+i\neq (k-1)M+j \\ 1\le l,k\le K,1\le i,j\le M}}\frac{|\bm e_{li}^H \bm e_{kj}|}{\|\bm e_{li}\|_2\cdot \|\bm e_{kj}\|_2} \\
&= \max_{\substack{(l-1)M+i\neq (k-1)M+j \\ 1\le l,k\le K,1\le i,j\le M}}\frac{|\bm s_l^H\bm s_k| |\bm d_i^H \bm d_j|}{\|\bm s_l\|_2\|\bm s_k\|_2\|\bm d_i\|_2\|\bm d_j\|_2}\\
&= \left\{
     \begin{array}{ll}
       \mu\{\bm S^T\}\mu\{\bm D^u\}, & \hbox{\(i\neq j,l\neq k\);} \\
       \mu\{\bm D^u\}, & \hbox{\(i\neq j, l=k\);} \\
       \mu\{\bm S^T\}, & \hbox{\(i=j, l\neq k\).}
     \end{array}
   \right.
\end{aligned}
\end{equation}
Notice that the mutual coherence is always smaller or equal to 1, i.e., \(\mu\{\bm S^T\}\le 1, \mu\{\bm D^u\}\le 1\). So we have \(\mu\{\bm E\}= \text{max}\{\mu\{\bm S^T\},\mu\{\bm D^u\}\}\) for \(1\le l,k\le K,1\le i,j\le M\).
\end{proof}
Theorem \ref{theorem: muE} indicates that to minimize \(\mu\{\bm E\}\), the larger one of \(\mu\{\bm S^T\}\) and \(\mu\{\bm D^u\}\) needs to be minimized.
Notice that \(\bm D^u\) is the sparsifying matrix which models the channel, and it has been designed before the channel estimation (see details in Section \ref{section: dictionary learning}). So during the channel estimation phase, \(\mu\{\bm D^u\}\) is fixed and could be small depending on which \(\bm D^u\) is used. The only way we can minimize \(\mu\{\bm E\}\) is by minimizing \(\mu\{\bm S^T\}\), which corresponds to \emph{design uplink pilots} such that \(\mu\{\bm S^T\}\) is small. We discuss different situations regarding to the length of the uplink training duration \(T^u\):

\subsubsection{\(T^u=1\)} when \(T^u=1\), \(\bm S^T=[s_1,\ldots, s_K]\in\mathbb{C}^{1\times K}\), so \(\mu\{\bm S^T\}=1\) for any \(\bm S^T\). This is the worst case since even we pick sparsifying matrix \(\bm D^u\) such that \(\mu\{\bm D^u\}=0\), we still have \(\mu\{\bm E\}=1\), i.e., there exist fully correlated columns.
No sparse recovery algorithm can succeed in this situation.
\subsubsection{\(T^u\ge K\)} when \(T^u\ge K\) we have \(\text{min}_{\bm S^T}\mu\{\bm S^T\}=0\), where the optimal \(\bm S^T\) has orthogonal columns, i.e., \(\bm s_l^H \bm s_k=0, \forall l\neq k\). So the optimal uplink pilots design is \(\bm S^*\bm S^T=\textbf{I}_K\). The orthogonal pilots among users in the same cell are typically assumed for the uplink channel estimation in multiuser Massive MIMO systems \cite{marzetta2006much, hoydis2013massive}.
\subsubsection{\(1<T^u<K\)} when \(1<T^u<K\), \(\bm S^T\in\mathbb{C}^{T^u\times K}\) is an overcomplete matrix. The famous welch bound indicates that
\begin{equation}\label{equ: welch bound}
\mu\{\bm S^T\} \ge \sqrt{\frac{K-T^u}{T^u(K-1)}}
\end{equation}
where equality holds if and only if \(\bm S^T = [\bm s_1,\ldots, \bm s_K]\) forms an \emph{equiangular tight frame} \cite{foucart2013mathematical}. Unfortunately, equiangular tight frame does not exist for any pair \(\{T^u, K\}\). In \cite{strohmer2003grassmannian}, the solution \(\bm S^T\) to the problem \(\text{min}_{\bm S^T}\mu\{\bm S^T\}\) is called \emph{Grassmannian frame}, and explicit construction of Grassmannian frame has been provided for some specific pairs \(\{T^u, K\}\). In general, the design of Grassmannian frames is challenging. Not only is the associated optimization problem difficult, but there is no general procedure for deciding when a frame solves the optimization problem unless it meets the Welch bound \cite{tropp2005designing}. In this paper, we design \(\bm S^T\) following the algorithm proposed in \cite{elad2007optimized} which targets an average measure of the mutual coherence. The algorithm calculates the Gram matrix of \(\bm S^T\) as \(\bm G=\bm S^*\bm S^T\), and set the average mutual coherence \(\mu_{t\%}\{\bm S^T\}\) such that the top \(t\%\) of \(|\bm G_{ij}|\) is greater than \(\mu_{t\%}\{\bm S^T\}\). The algorithm then shrinks those large \(|\bm G_{ij}|\) by some down-scaling factor \(\gamma\) to have \(\tilde{\bm G}_{ij} = \gamma \bm G_{ij}\), and keeps the small ones unchanged. The estimated \(\hat{\bm S}^T\) is the solution of \(\text{min}_{\bm S^T}\|\tilde{\bm G}-\bm S^*\bm S^T\|_F^2\), which is solved by SVD of \(\tilde{\bm G}\). Then new \(\bm G\) is calculated and such procedure is iteratively executed until some stopping rule is satisfied. By iterative shrinkage of those large \(|\bm G_{ij}|\), the \(\mu\{\bm S^T\}\) is also reduced. It has been shown in \cite{elad2007optimized} that the algorithm practically converged and the resulted \(\bm S^T\) can lead to better performance for the sparse recovery problem like (\ref{equ: uplink sparse reocvery}). Interested readers are referred to \cite{tropp2005designing, elad2007optimized} for more details.

Next we consider minimizing \(\mu\{\bm E_{\Lambda}\}\). Denote \(\Lambda_k=\text{supp}(\bm \beta_k)\), then \(\bm E_{\Lambda}\) can be written as
\begin{equation}\label{equ: E_Lambda}
\bm E_{\Lambda}=\left[
                      \begin{array}{cccc}
                        \bm s_1\otimes \bm D_{\Lambda_1}^u & \bm s_2\otimes \bm D_{\Lambda_2}^u & \ldots & \bm s_K\otimes \bm D_{\Lambda_K}^u \\
                      \end{array}
                    \right].
\end{equation}
We take a simple example in the following to see how \(\Lambda_k\) can affect the recovery performance when \(T^u<K\). Assume \(K=2\), and both \(\Lambda_1\) and \(\Lambda_2\) are known \textit{a priori} with \(|\Lambda_1| = |\Lambda_2| =1\). Let \(T^u = 1\), so \(s_1\) and \(s_2\) are scalars. If \(\Lambda_1 \cap \Lambda_2 = \varnothing\), then \(\text{rank}(\bm E_{\Lambda})=2\) and we can robustly recovery \(\bm b^u_{\Lambda}\) from \(\bm y^u = \bm E_{\Lambda}\bm b^u_{\Lambda} +\bm w^u\) when the correlation of \(\bm D_{\Lambda_1}^u\) and \(\bm D_{\Lambda_2}^u\) is small. However, if \(\Lambda_1\) is overlapped with \(\Lambda_2\), which in this example means \(\Lambda_1 = \Lambda_2\), then \(\text{rank}(\bm E_{\Lambda})=1\) and we have \(\bm y^u = \bm D_{\Lambda_1}^u(s_1 b^u_{\Lambda_1} + s_2 b^u_{\Lambda_2} ) +\bm w^u\) making recovery of \(b^u_{\Lambda_1}\) and \(b^u_{\Lambda_2}\) impossible. In this case, \(T^u\ge 2\) is required to estimate \(b^u_{\Lambda_1}\) and \(b^u_{\Lambda_2}\). This example motivates how the non-overlapping supports of different users can help sparse recovery when \(T^u<K\), as formally shown in the following corollary.

\begin{corollary}\label{corollary: muE_Lambda}
Given \(\bm E_{\Lambda}\) in (\ref{equ: E_Lambda}) and the mutual coherence defined in (\ref{equ: mutual coherence}), then \(\mu\{\bm E_{\Lambda}\}= \mu\{\bm D^u\}\) if \(\Lambda_l\cap \Lambda_k = \varnothing, \forall l\neq k\).
\end{corollary}
\begin{proof}
Following the same steps as in the proof of Theorem \ref{theorem: muE}, then the condition \(\Lambda_l\cap \Lambda_k = \varnothing, \forall l\neq k\) implies \(i\neq j\) for \(\bm e_{li}\) and \(\bm e_{kj}\). So \(\mu\{\bm E_{\Lambda}\}= \text{max}\{\mu\{\bm S^T\}\mu\{\bm D^u\},\mu\{\bm D^u\}\} = \mu\{\bm D^u\}\) since \(\mu\{\bm S^T\}\le 1\) for any \(\bm S^T\).
\end{proof}
Comparing to \(\mu\{\bm E\}\), \(\mu\{\bm E_{\Lambda}\}\) is no longer depending on \(\mu\{\bm S^T\}\) when the support sets of different users are \emph{non-overlapping}. So even if \(\mu\{\bm S^T\}\) is large
\footnote{In the case when the number of users is much larger than the number of pilot symbols, i.e., \(K\gg T^u\), we have \(\mu\{\bm S^T\}\ge 1/\sqrt{T^u}\) from the welch bound in (\ref{equ: welch bound}).},
\(\mu\{\bm E_{\Lambda}\}\) can still be small if \(\mu\{\bm D^u\}\) is small. This result sheds light on how \emph{user scheduling} can affect the performance of channel estimation. If given prior knowledge of \(\Lambda_k = \text{supp}(\bm \beta_k)\) 
\footnote{Such prior knowledge can possibly be obtained using some kind of control information, or from previous estimate of \(\bm \beta_k\) when users are slowly moving.}
, we can schedule users whose supports satisfy \(\Lambda_l\cap \Lambda_k = \varnothing, \forall l\neq k\), which will lead to smaller \(\mu\{\bm E_{\Lambda}\}\) and better channel estimation. 
This result is consistent with \cite{yin2013coordinated}, which shows that in a \textit{multi-cell} network user interference can be eliminated by simple MMSE channel estimation when the AOA of the desired user has no overlap with AOAs of interfering users.  
Interestingly, authors in \cite{rao2014distributed,liu2017closed} suggest to schedule users with overlapped supports for the \textit{downlink} channel estimation, since it can be formulated into a joint sparse recovery problem which exploits the common support information among users.
For the uplink training, in contrast, common support increases mutual coherence when \(T^u<K\), causing decreased performance when applying sparse recovery algorithm for channel estimation.

\section{Dictionary Learning based Channel Model}\label{section: dictionary learning}

\subsection{Predefined Sparsifying Matrix}
\label{subsection: Sparsifying Basis/Dictionary}
In both compressive sensing based downlink channel estimation and sparse recovery based uplink channel estimation, the key assumption is that the channel can be represented in the form of \(\bm h=\bm D\bm \beta\), where \(\bm \beta\) is a sparse vector
\footnote{The concept in this section applies to both downlink and uplink channels, so we drop the superscript \(d\) and \(u\) to simplify the notation.}.
The existing works \cite{bajwa2010compressed, rao2014distributed, sayeed2002deconstructing, tse2005fundamentals, gao2015spatially, shen2016compressed, zhou2016sparse} which consider such a sparse representation typically use a normalized square DFT matrix as the sparsifying matrix when an \emph{uniform linear array} (ULA) is employed, i.e., \(\bm h=\bm F\bm \beta\), where
\begin{equation}\label{equ: square DFT}
\begin{aligned}
&\bm F=\begin{bmatrix}
       \bm f(-\frac{1}{2}) & \bm f(-\frac{1}{2}+\frac{1}{N}) & \ldots & \bm f(\frac{1}{2}-\frac{1}{N}) \\[0.3em]
     \end{bmatrix}\in\mathbb{C}^{N\times N},\\
&\bm f(\psi)=\frac{1}{\sqrt{N}}[
       1 , e^{j2\pi\psi} , \ldots , e^{j2\pi\psi\cdot (N-1)}
     ]^T.\\
\end{aligned}
\end{equation}
Such a model is also known as the ``virtual channel model" which transforms spatial channel response into the angular domain \cite{bajwa2010compressed, sayeed2002deconstructing, tse2005fundamentals}.
Notice that the column \(\bm f(\psi)\) has the same structure as the array response \(\bm a(\theta)\) in (\ref{equ: array response}), and \(\psi\) can be related to \(\theta\) through \(\psi=d\text{sin}(\theta)/\lambda\), indicating the validity of the DFT matrix .
However, in practice signals come from arbitrary directions, so \(\psi=d\text{sin}(\theta)/\lambda\) rarely resides on the DFT bins \(\{-\frac{1}{2}, -\frac{1}{2}+\frac{1}{N}, \ldots, \frac{1}{2}-\frac{1}{N}\}\), leading to the ``leakage" effect. Moreover, as demonstrated in Section \ref{section: system model}, for each scattering cluster the signals' subpaths often span an angular spread, resulting even more leakage. So for practical channels, there will be a lot of nonzero elements in \(\bm \beta\) when we apply the representation \(\bm h=\bm F\bm \beta\), making the sparse assumption invalid.

To achieve a better sparse representation, our first suggestion in the same realm of ``predefined matrix" for ULA is to apply the \emph{overcomplete} DFT matrix \(\tilde{\bm F}\), which has the form
\begin{equation}\label{equ: overcomplete DFT}
\begin{aligned}
&\tilde{\bm F}=\begin{bmatrix}
       \bm f(-\frac{1}{2}) & \bm f(-\frac{1}{2}+\frac{1}{M}) & \ldots & \bm f(\frac{1}{2}-\frac{1}{M}) \\[0.3em]
     \end{bmatrix}\in\mathbb{C}^{N\times M}.\\
\end{aligned}
\end{equation}
The columns of \(\tilde{\bm F}\) has the same structure \(\bm f(\psi)\), but the angular domain is sampled (in the sense of \(\psi\)) more finely, i.e., \(M>N\). The overcomplete DFT matrix introduces redundancy to the square DFT matrix, which improves the flexibility of representing the signal as well as the capability of inducing sparsity. When a URA with \(N_1\) vertical antennas and \(N_2\) horizontal antennas is applied, \(\tilde{\bm F}\) is constructed as the kronecker product of two overcomplete DFT matrix such that
\begin{equation}
\tilde{\bm F} = \tilde{\bm F}_h \otimes \tilde{\bm F}_v
\end{equation}
where \(\tilde{\bm F}_h\) and \(\tilde{\bm F}_v\) are \(N_2\times M_2\) and \(N_1\times M_1\) overcomplete DFT matrices as given in (\ref{equ: overcomplete DFT}).
In Section \ref{section: simulation}, we show experimentally how this simple extension to the overcomplete DFT matrix can greatly improve the performance.

Although the overcomplete DFT matrix can alleviate the leakage effect to some extent, both \(\bm F\) and \(\tilde{\bm F}\) suffer from performance loss due to their inability to adapt to the real channels. Firstly, since \(\tilde{\bm F}\) and \(\bm F\) are predefined and independent of the specific cell properties, they lose the ability to more \emph{efficiently} represent the channel by exploring \emph{cell specific } characteristics. For example both \(\tilde{\bm F}\) and \(\bm F\) uniformly sample the \(\psi\) domain, but for a specific cell it is possible that no signals may be received from some directions, then the columns in \(\tilde{\bm F}\) and \(\bm F\) corresponding to those directions will never be used. On the other hand, for directions corresponding to those fixed location scattering clusters, finer angular sampling can lead to a reduced leakage. Since those fixed location scattering clusters can be seen by many different users, such finer sampling can lead to more sparse representation for many users. Secondly, predefined matrices also lose the ability to \emph{robustly} represent the channel. They assume ideal mathematical models of channel responses, e.g., far-field plane wave, equal antenna gain and antenna spacing, etc., which are not robust to any propagation model mismatch or antenna array uncertainty.

\subsection{Dictionary Learning}
In this paper, we propose a \emph{dictionary learning based channel model} (DLCM) which learns an overcomplete dictionary.
During the learning process, the sparse representation is encouraged by the optimization function.
Furthermore, the dictionary learning process adapts the channel model to the channel measurements collected in the cell, which contain the specific cell characteristics\footnote{The channel measurements describe the effect of scattering clusters on the transmitted electromagnetic waves and antenna array. The underlying structure of channel measurements collected in a specific cell can reflect the cell specific properties regarding to both scattering clusters and the antenna array.}.
Notice that when the knowledge of the underlying physical generation scheme of the channel is imperfect or even incorrect, for example, antenna gains and locations are different from the nominal values, or there exist near-field scattering clusters, the predefined matrix is no longer accurate and may cause severe performance degradation. However, the learned dictionary does not have any predefined structural constraints and is able to tune its own structure to adapt to the channel measurements, which leads to a more robust channel representation.
The insight behind the sparse dictionary learning is that the high dimensional data (channel response in our case) usually has some structure correlated in some dimensions, and the true degrees of freedom that generate the data is usually small. So by learning from large amount of data, we are able to recover useful underlying structures or models, which make the representation of the data more efficient for the desired application. In our situation, one could view this as big data analytics applied to the physical layer.

From now, we denote \(\bm D\in\mathbb{C}^{N\times M}\) as the learned dictionary from channel measurements. To benefit from the flexibility of overcompleteness, we let \(N<M\). Assuming we collect \(L\) channel measurements as the training samples in a specific cell, the goal is to learn \(\bm D\) such that for all the channel responses \(\bm h_i, i=1,\ldots,L,\) they can be approximated as \(\bm h_i\approx \bm D\bm \beta_i\). The algorithms should be able to address both model fitting \(\|\bm h_i-\bm D\bm \beta_i\|_2\) (robustness), and encourage small \(\|\bm \beta_i\|_0\) (efficiency) for the sparse representation. If we constrain the model mismatch error of each channel response to be bounded by \(\eta\), then the dictionary learning can be formulated as
\begin{equation}\label{equ: DL_BPl0}
\begin{aligned}
&\min_{\substack{\bm D\in\mathcal{C}\\ \bm\beta_1,\ldots,\bm\beta_L }} \frac{1}{L}\sum_{i=1}^L\|\bm\beta_i\|_0\quad
\text{subject to }\|\bm h_i-\bm D\bm \beta_i\|_2\le \eta, \forall i
\end{aligned}
\end{equation}
where the constraint set \(\mathcal{C}\) is defined as
\begin{equation}
\mathcal{C}=\{\bm D\in\mathbb{C}^{N\times M}, \text{ s.t. } \|\bm D_{\cdot j}\|_2\le 1,\forall j=1,\ldots,M \}
\end{equation}
in order to prevent the scaling ambiguity between columns of \(\bm D\) and corresponding elements in \(\bm \beta\).
The solved \(\bm D\) in (\ref{equ: DL_BPl0}) leads to the sparsest representation in the sense of representing all collected channel measurements within the model mismatch tolerance \(\eta\).

To solve the dictionary learning problems (\ref{equ: DL_BPl0}), block coordinate descent framework has been applied where each iteration includes alternatively minimizing with respect to either \(\bm D\) or \(\bm \beta_i, \forall i\), while keeping the other fixed \cite{engan1999method, kreutz2003dictionary, aharon2006svd}. When \(\bm D\) is fixed, optimizing \(\bm \beta_i, \forall i\) is decoupled and each of them is a sparse recovery problem, which can be solved by any sparse recovery algorithm. When we fix \(\bm \beta_i, \forall i\) and solve for \(\bm D\), many dictionary learning algorithm can be applied \cite{engan1999method, kreutz2003dictionary, aharon2006svd}. The convergence of the iteration depends on the specific sparse recovery algorithm and dictionary update algorithm, and to the best of our knowledge, no general guarantees have been provided. Interested readers are referred to \cite{tseng2001convergence, aharon2006svd} for some discussion about the convergence under specific assumptions. Notice that in our scenario, there exists \textit{no} ``true" dictionary that generates the channel. Because each channel response combines signals coming from both fixed location scattering clusters and user location dependent scattering clusters, where the latter depends on arbitrary user's location. So the goal of the dictionary learning here is not to identify any true dictionary \cite{remi2010dictionary}, but to find an efficient and robust channel representation. For the purpose of this paper, we show experimentally in Section \ref{section: simulation} that the learned dictionary can improve the performance in terms of both sparse representation and channel estimation.

\subsection{Discussion}
We make some comments relative to the practical implementation of the dictionary learning process. To learn a comprehensive dictionary for users located in any place of the cell, we need to collect channel measurements from all locations in a specific cell, i.e., \emph{cell specific samples}, based on an extensive measurement campaign. The learned dictionaries will only be used for this specific cell. At this stage there is not much concern about reducing training and feedback overhead and one would like to collect channel measurements as accurately and as extensively as possible, since large amount of channel samples will prevent the learning algorithm from overfitting. For example, one can perform conventional channel estimation using more training pilots, larger transmitted power and more sophisticated equipment. Fortunately, such channel measurements collection and dictionary learning process is \emph{offline}, and assumed to be done at the cell deployment stage. Due to the non-convex learning process, it is possible that the learned dictionary converges to local optima. Starting from a reasonably good initial point, for example an overcomplete DFT matrix, can help avoid such local optima and promote quicker convergence.

The learned dictionary \(\bm D^d\) and \(\bm D^u\) are stored at the BS for use during downlink and uplink channel estimation. It is straightforward for the uplink since the channel estimation is performed at the BS, and we can directly applied \(\bm D^u\) in (\ref{equ: uplink sparse reocvery}).
In downlink channel estimation, users feed back the received measurements \(\bm y^d\) to the BS and the channel is estimated at the BS using (\ref{equ: l0 minimization}) with the learned dictionary \(\bm D^d\). This is different from the conventional channel estimation where users estimate the channel and feed back the channel state information to the base station. The scheme of feeding back \(\bm y^d\) has been proposed in previous works \cite{rao2014distributed, gao2015spatially, zhou2016sparse, liu2017closed}, which has several advantages: firstly the sparse recovery algorithms (channel estimation) can be complex so it is preferably done at the BS thus saving energy for UE. Secondly, \(\bm y^d\) has dimension \(T^d\) which is much less than the channel dimension \(N\) in Massive MIMO system, so it also reduces feedback overhead which is now only proportional to the channel sparsity level. Furthermore, for the downlink channel estimation, making the learned dictionary available to all users involves significant overhead in storage at UE and also conveyance of dictionary. By feeding back \(\bm y^d\) only the BS needs to know the dictionaries. In this work, perfect uplink feedback is assumed for simplicity.

\section{Uplink/Downlink Joint Dictionary Learning}\label{section: joint dictionary learing}

\subsection{Motivation}
In compressive sensing based downlink channel estimation, larger training period \(T^d\) leads to better recovery performance since more information about the downlink channel is collected. However, larger \(T^d\) also means more downlink resources for channel estimation and leaves less time for actual data transmission. This motivates our search for alternative information sources that can facilitate downlink channel estimation.
For this we draw inspiration from TDD systems,
where through channel reciprocity the uplink channel estimate provides  downlink channel information \cite{marzetta2006much, hoydis2013massive}.
In FDD system we do not have such channel reciprocity because uplink and downlink transmission are operated in different frequency bands. However, if the duplex distance is not large, i.e., the frequency difference between the uplink and the downlink is not large, a looser and more abstract form of reciprocity is possible and appropriate. For instance, it is reasonable to assume the AOA of signals in the uplink transmission is the same as the AOD of signals in the downlink transmission \cite{zetterberg1995spectrum, paulraj1997space,  hugl2002spatial,hugl1999downlink,3gpp.25.996}.
In other words directions of signal paths are invariant to carrier frequency shift. In \cite{hugl2002spatial}, congruence of the directional properties of the uplink and the downlink channel is observed experimentally, where the dominant uplink/downlink directions of arrival (DOA) show only a minor deviation, and the uplink/downlink azimuth power spectrums (APS) have a high correlation.
This indicates that in the spatial channel model (\ref{equ: SCM_downlink}),
\(\alpha_{il}\) in the uplink is different and uncorrelated from \(\alpha_{il}\) in the downlink due to the frequency separation, but both links share the same \(N_c, N_s\) and \(\Omega_{il}\). So when we treat \(\bm h^u\) and \(\bm h^d\) as a whole, they appear to be uncorrelated. But if we are able to resolve them finely in the angular domain, which indeed can be achieved by the large antenna array, they will show the common spatial structure which can be regarded as the reciprocity in the angular domain.
Furthermore, the directions in the angular domain are closely related to the locations of nonzero entries in the sparse coefficients. So the reciprocity in the angular domain translates to the same locations of nonzero entries in \(\bm \beta^u\) and \(\bm \beta^d\), i.e., \(\text{supp}(\bm \beta^u)=\text{supp}(\bm \beta^d)\). Consequently, if we know \(\bm h^u\), and utilize for the downlink channel estimation  the common support information \(\text{supp}(\bm \beta^u)=\text{supp}(\bm \beta^d)\), we have critical information about \(\bm h^d\) and can obtain better downlink channel estimates without increasing the training overhead.

\subsection{Joint Dictionary Learning}
\label{subsection: Upllink Donwlink joint dictionary learning}
Based on the DLCM in the previous section, we propose a \emph{joint} dictionary learning process where \(\bm D^u\) and \(\bm D^d\) are learned jointly with the constraint on the support, i.e., \(\text{supp}(\bm \beta^u)=\text{supp}(\bm \beta^d)\). In order to enforce such constraint, we collect channel samples \(\{\bm h^u_i, \bm h^d_i\}\) in pair. Each pair of samples is measured at the same UE location, so the assumption of the same AOA/AOD is valid. The joint dictionary learning can be formulated as
\begin{equation}\label{equ: JDL_BPl0}
\begin{aligned}
&\min_{\substack{\bm D^u\in\mathcal{C}, \bm\beta_1^u,\ldots,\bm\beta_L^u\\
\bm D^d\in\mathcal{C}, \bm\beta_1^d,\ldots,\bm\beta_L^d }} \enspace \frac{1}{L}\sum_{i=1}^L\|\bm\beta_i^u\|_0+\|\bm\beta_i^d\|_0\\
&\text{subject to} \quad \|\bm h_i^u-\bm D^u\bm \beta_i^u\|_2\le \eta^u, \enspace \|\bm h_i^d-\bm D^d\bm \beta_i^d\|_2\le \eta^d,\\
&\quad \quad \quad \quad \quad  \text{supp}(\bm \beta_i^u)=\text{supp}(\bm \beta_i^d), \enspace \forall i
\end{aligned}
\end{equation}
which is very similar to the dictionary learning problem as shown in (\ref{equ: DL_BPl0}), except for the constraint \(\text{supp}(\bm \beta_i^u)=\text{supp}(\bm \beta_i^d)\). This constraint is important since it builds the connection between the uplink and downlink channel responses, which will be utilized in the joint channel estimation. 

To solve the joint dictionary learning, we minimize (\ref{equ: JDL_BPl0}) iteratively, i.e., we fix \(\bm D^u, \bm D^d\) and solve for \(\bm\beta_i^u, \bm \beta_i^d,\forall i\), and then fix \(\bm\beta_i^u, \bm \beta_i^d,\forall i\) and solve for \(\bm D^u, \bm D^d\). Notice that when \(\bm\beta_i^u, \bm \beta_i^d,\forall i\) are fixed, the solution of \(\bm D^u\) and \(\bm D^d\) are decoupled, and can be optimized independently using any of dictionary learning algorithms \cite{engan1999method, kreutz2003dictionary, aharon2006svd}. When \(\bm D^u, \bm D^d\) are fixed, different pairs of \(\{\bm \beta^u_i, \bm \beta^d_i\}\) are decoupled. But for each of the pair, they are coupled through the constraint \(\text{supp}(\bm \beta_i^u)=\text{supp}(\bm \beta_i^d)\), and need to be solved jointly. Algorithms aiming to solve joint sparse recovery have been proposed in previous works,
such as OMP like algorithm
\cite{baron2005distributed}, \(\ell_1\) norm algorithm \cite{yuan2006model}, reweighted \(\ell_p\) norm algorithm  \cite{ding2015joint} and sparse Bayesian learning algorithm \cite{ji2009multitask}. It has been shown that joint recovery can lead to more accurate results compared to  independent recovery.
In this paper, we consider a group \(\ell_1\) formulation which is similar to the group-lasso in \cite{yuan2006model} to solve the joint sparse recovery problem. By forming
\begin{equation}
\bm h=\left[
        \begin{array}{c}
          \bm h^d_i \\
          \bm h^u_i \\
        \end{array}
      \right],
\bm \beta=\left[
           \begin{array}{c}
             \bm \beta^d_i \\
             \bm \beta^u_i \\
           \end{array}
         \right],
\bm G=\left[
        \begin{array}{cc}
          \bm D^d & \textbf{0}_{N\times M} \\
          \textbf{0}_{N\times M} & \bm D^u \\
        \end{array}
      \right]
\end{equation}
the joint sparse recovery of \(\bm \beta^d_i, \bm \beta^u_i\) can be written as
\begin{equation}\label{equ: groupbpdn_h}
\begin{aligned}
&\min_{\bm \beta}\sum_{j=1}^M\|\bm\beta\|_{\bm K_j}\quad
\text{subject to } \|\bm h-\bm G\bm \beta\|_2\le\eta
\end{aligned}
\end{equation}
where \(\|\bm\beta\|_{\bm K_j}=(\bm \beta^H \bm K_j\bm\beta)^{1/2}\), \(\bm K_j=\text{diag}([\bm e_j^T \enspace \bm e_j^T]^T)\) is the group kernel, where \(\bm e_j\in\mathbb{R}^{M\times 1}\) is the standard basis with \(1\) in the \(j\)-th location and \(0\) elsewhere. The group kernel gathers the \(j\)-th element in \(\bm \beta_i^d\) and the \(j\)-th element in \(\bm \beta_i^u\) into the same group.  
The cost function in (\ref{equ: groupbpdn_h}) is a \(\ell_2/\ell_1\) norm of \(\bm \beta\) similar to the \(\ell_p/\ell_1\) norm in \cite{cotter2005sparse}, which encourages all the elements in the same group to be zero or nonzero simultaneously, and the total number of nonzero groups to be small. By applying this group \(\ell_1\) framework, we enforce the constraint of \(\text{supp}(\bm \beta_i^u)=\text{supp}(\bm \beta_i^d)\) and encourage a sparse representation.

\subsection{Joint Channel Estimation}
After learning \(\bm D^u, \bm D^d\), we have the joint uplink and downlink sparse channel representation as \(\bm h^u\approx \bm D^u\bm \beta^u\) and \(\bm h^d\approx \bm D^d\bm\beta^d\). The goal is to utilize uplink training to help improving the performance of the downlink channel estimation, by using the constraint \(\text{supp}(\bm \beta^u)=\text{supp}(\bm \beta^d)\).
Consider the single user case in (\ref{equ: uplink training}), we have \(\bm Y^u = \bm h^u \sqrt{\rho^u T^u} \bm s^T +\bm W^u\). Denote \(\bm y^u = \bm Y^u (\sqrt{\rho^u T^u}\bm s^T)^{\dagger}\) and \(\bm w^u = \bm W^u (\sqrt{\rho^u T^u}\bm s^T)^{\dagger}\), we have
\begin{equation}
\bm y^u = \bm h^u +\bm w^u = \bm D^u\bm \beta^u +\bm w^u.
\end{equation}
Combined with the downlink training in (\ref{equ: downlink compressed training}), the compressed channel estimation can be formulated as
\begin{equation}
\begin{aligned}
&\{\hat{\bm \beta}^u,\hat{\bm \beta}^d\}=\arg\min_{\bm \beta^u, \bm \beta^d} \|\bm \beta^u\|_0+\|\bm \beta^d\|_0 \\
&\text{subject to}\quad \|\bm y^u -\bm D^u\bm \beta^u\|_2\le \epsilon^u, \enspace
\|\bm y^d -\bm A\bm D^d\bm \beta^d\|_2\le \epsilon^d,\\
&\quad\quad\quad\quad\quad \text{supp}(\bm \beta^u)=\text{supp}(\bm \beta^d).
\end{aligned}
\end{equation}
And the uplink and downlink channel can be estimated as \(\hat{\bm h}^u=\bm D^u\hat{\bm \beta}^u\) and \(\hat{\bm h}^d=\bm D^d\hat{\bm \beta}^d\). Again, we face the same joint sparse recovery problem with structure constraint \(\text{supp}(\bm \beta^u)=\text{supp}(\bm \beta^d)\) as in the joint dictionary learning problem. We utilize the same group \(\ell_1\) algorithm as in (\ref{equ: groupbpdn_h}) as following
\begin{equation}\label{equ: groupbpdn_y}
\begin{aligned}
&\min_{\bm \beta}\sum_{j=1}^M\|\bm\beta\|_{\bm K_j}\quad
\text{subject to } \|\bm y-\bm G\bm \beta\|_2\le\epsilon
\end{aligned}
\end{equation}
where now we have
\begin{equation}\label{equ: yandG}
\bm y=\left[
        \begin{array}{c}
          \bm y^d \\
          \tau \bm y^u \\
        \end{array}
      \right],
\bm G=\left[
        \begin{array}{cc}
          \bm A\bm D^d & \textbf{0}_{T\times M} \\
          \textbf{0}_{N\times M} & \tau \bm D^u \\
        \end{array}
      \right]
\end{equation}
and the same definition of \(\bm K_j\). Notice that the norm of columns in \(\bm A\bm D^d\) can be much larger than the norm of columns in \(\bm D^u\) when \(\rho^d\) is large, which deemphasizes the role of uplink training in the noisy situation. So a constant \(\tau\) is multiplied to make the columns of \(\bm G\) to have similar norms. The solved \(\bm \beta\) has the form of \(\bm \beta = [(\bm \beta^d)^T (\bm \beta^u)^T]^T\).  By joint sparse recovery of \(\bm \beta^u, \bm \beta^d\), we are able to achieve improved  downlink channel estimates with the help of uplink training measurements. Notice the dimension of \(\bm y^u\) is \(N\) while dimension of \(\bm y^d\) is \(T^d\). In the Massive MIMO system where \(N\gg T^d\), the uplink training actually has larger number of measurements, which is beneficial for the sparse recovery algorithm.
We can also improve the signal to noise ratio of the uplink received signal by increasing the uplink training period \(T^u\).
Due to the constraint \(\text{supp}(\bm \beta^u)=\text{supp}(\bm \beta^d)\), \(\bm y^u\) and \(\bm y^d\) can regularize each other to achieve better recovery performance compared to independent recovery. More importantly, the performance of the downlink compressed channel estimation is improved without increasing the downlink training period \(T^d\).

Similar to the DLCM, there is a joint dictionary learning phase and a joint channel estimation phase.
During the joint dictionary learning phase, a large amount of channel measurements need to be collected as training samples. Each pair of uplink/downlink channel measurements has to be collected at the \emph{same user location}, in order to guarantee the same AOA/AOD for the uplink and downlink. This requirement is important since for each pair of \(\{\bm h^u_i,\bm h^d_i\}\) the learning process has the constraint \(\text{supp}(\bm \beta^u_i)=\text{supp}(\bm \beta^d_i)\). The joint dictionary learning is implemented when the cell is installed, and the learned \(\bm D^u, \bm D^d\) are stored at the base station. In the channel estimation phase, the BS transmits downlink pilots while the UE transmits uplink pilots, then the UE feeds back the received signal. The joint channel estimation is performed at the BS, from which the uplink and downlink channel state information is obtained.

\section{Simulation Results}\label{section: simulation}
In the simulation, we test both 2D and 3D channel model. For 2D channel, we assume the BS is equipped with an ULA with 100 antennas and each UE has a single antenna. The channel is generated using parameters from non line-of-sight (NLOS) Urban Macro scenario in \cite{3gpp.25.996}. Since the learned dictionary depends on the cell characteristics, we generate cell specific scattering clusters following the principles of Geometry-Based Stochastic Channel Model (GSCM) \cite{molisch2003geometry}. At the beginning of the simulation, \(21\) fixed location scattering clusters are uniformly generated in a cell with radius \(900\) meters and \(\theta\in[-\frac{\pi}{2},\frac{\pi}{2}]\), and then kept constant for the simulation of both dictionary learning and channel estimation. The user's location is also randomly and uniformly generated. For each channel response between the BS and the UE, it consists AOA/AOD of multi-paths from \(3\) fixed location scattering clusters which are closest to the UE, and \(1\) user-location dependent scattering cluster which is generated according to \cite{3gpp.25.996} based on the UE's location. All the other parameters, e.g., angular spread, delay spread, and path power, are all generated following \cite{3gpp.25.996}. The AOA/AOD values are identical between the uplink and downlink, while the phases of subpaths are random and uncorrelated \cite{3gpp.25.996}.
For 3D channel, the BS is assumed to be equipped with a \(10\times 10\) URA and the UE with a \(3\times 3\) URA. The channel is generated following the NLOS UMi-Street Canyon scenario in \cite{3gpp.38.901}, where the carrier frequency is assumed to be \(28\) GHz. The generation of cell specific clusters is similar to the 2D model, except that the cell radius is \(200\) meters with \(\theta\in[0,\pi]\), \(\phi\in[0,\pi]\), and each cluster has a height \(h\in[0.5, 30]\) while \(h_{\text{BS}}=10\) m and \(h_{\text{UE}}=1.5\) m, so elevation angles ZOA/ZOD can be calculated. Since the carrier frequency is \(28\) GHz, each channel response consists only \(1\) fixed location scattering cluster and \(1\) user-location dependent scattering cluster consistent with the small number of scattering clusters at the millimeter wave (mmWave) frequency \cite{samimi20163}.

Two kinds of antenna array are considered at the BS. The first is the perfectly calibrated antenna array, i.e., equal spacing \(d = \lambda/2\) between antenna elements and equal antenna gains as $1$. In the second case, there exist \emph{antenna uncertainties} in the form of \emph{unknown but fixed} calibration errors, where the antenna spacing and gains are deviating from the nominal values. We generate them as following: in \(100\) antennas, 20 antennas have gains \(1+e\) while the other 80 have gains 1. The \(e\sim\mathcal{N}(0,0.1)\), and if \(1+e>1.2\) or \(1+e<0.8\), then the gain is set to be 1.2 or 0.8. Among 99 antenna spacings, there are 20 having values \(d=(1+v)\lambda/2\) where \(v\sim\mathcal{N}(0,0.1)\). If \(1+v>1.2\) or \(1+v<0.8\), then the spacing is set to be \(1.2\lambda/2\) or \(0.8\lambda/2\). The rest of antenna spacings are \(d=\lambda/2\). After the antenna gains and spacings are generated, they are fixed in the whole simulation of dictionary learning and channel estimation.

For the dictionary learning, \(L=10000\) channel responses are generated, and for each channel responses the UE is randomly and uniformly located in the cell with at least 300 meters (60 meters for the mmWave scenario) from the BS. K-SVD
\cite{aharon2006svd} combined with \(\ell_1\) or group \(\ell_1\) algorithm (implemented using SPGL1 toolbox \cite{van2008probing}) 
are applied. Unless otherwise indicated, the dictionary is learned from the true channel responses without accounting for any measurement noise
\footnote{We should emphasize it is an ideal assumption and our results are only to prove the concept of using dictionary learning for channel estimation. In the simulation, we provide an experimental study to show the effect of inaccurate training channel measurements for the dictionary learning.}.
We compare \(100\times 400\) learned dictionary \(\bm D\) (DLCM) with \(100\times 100\) DFT matrix \(\bm F\) (DFT) and \(100 \times 400\) overcomplete DFT matrix \(\tilde{\bm F}\) (ODFT).

\subsection{Sparse Representation Using DLCM}
\begin{figure*}[!t]
\centering
\begin{subfigure}[]{
\includegraphics[width=0.47\textwidth]{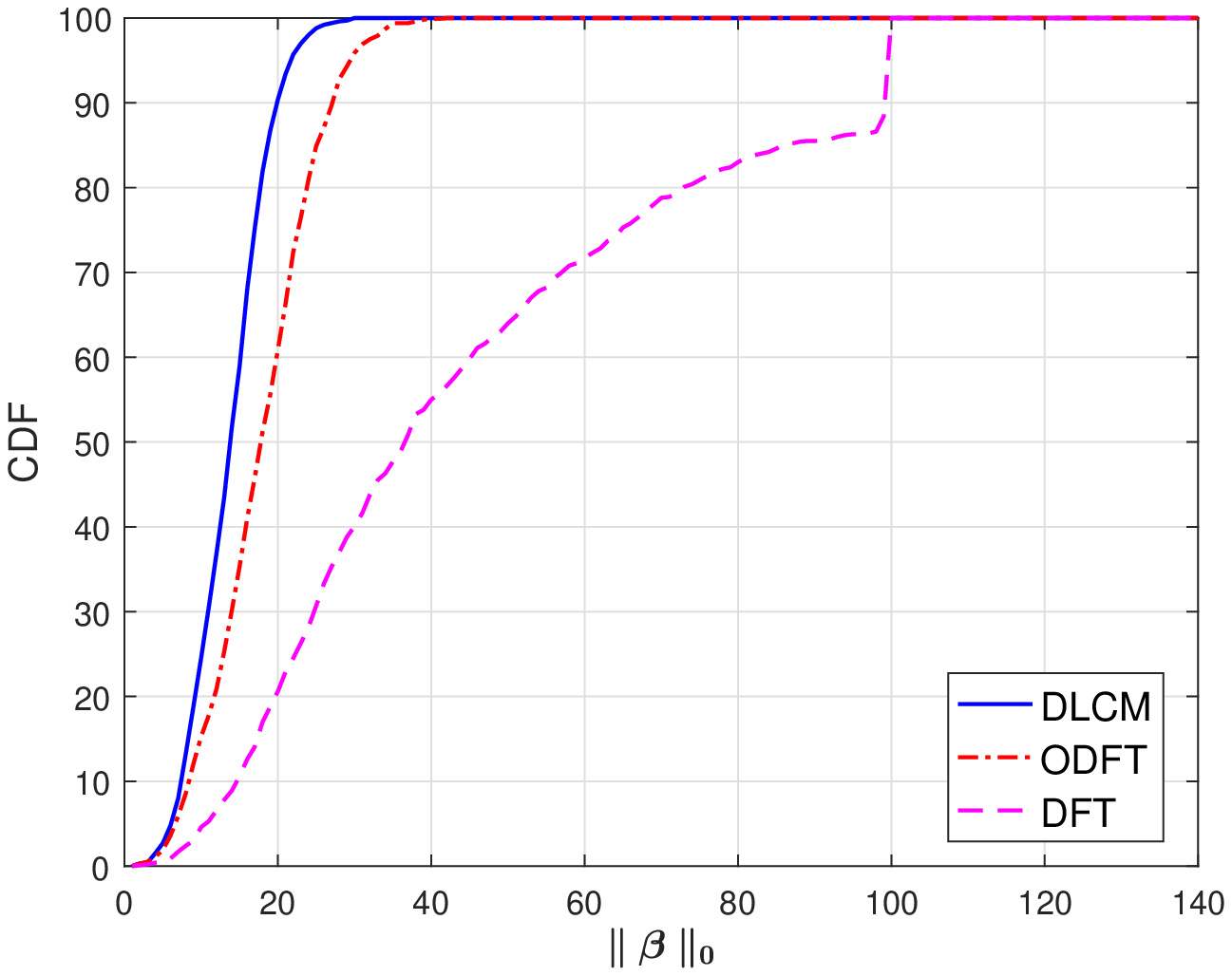}
}
\end{subfigure}
\begin{subfigure}[]{
\includegraphics[width=0.47\textwidth]{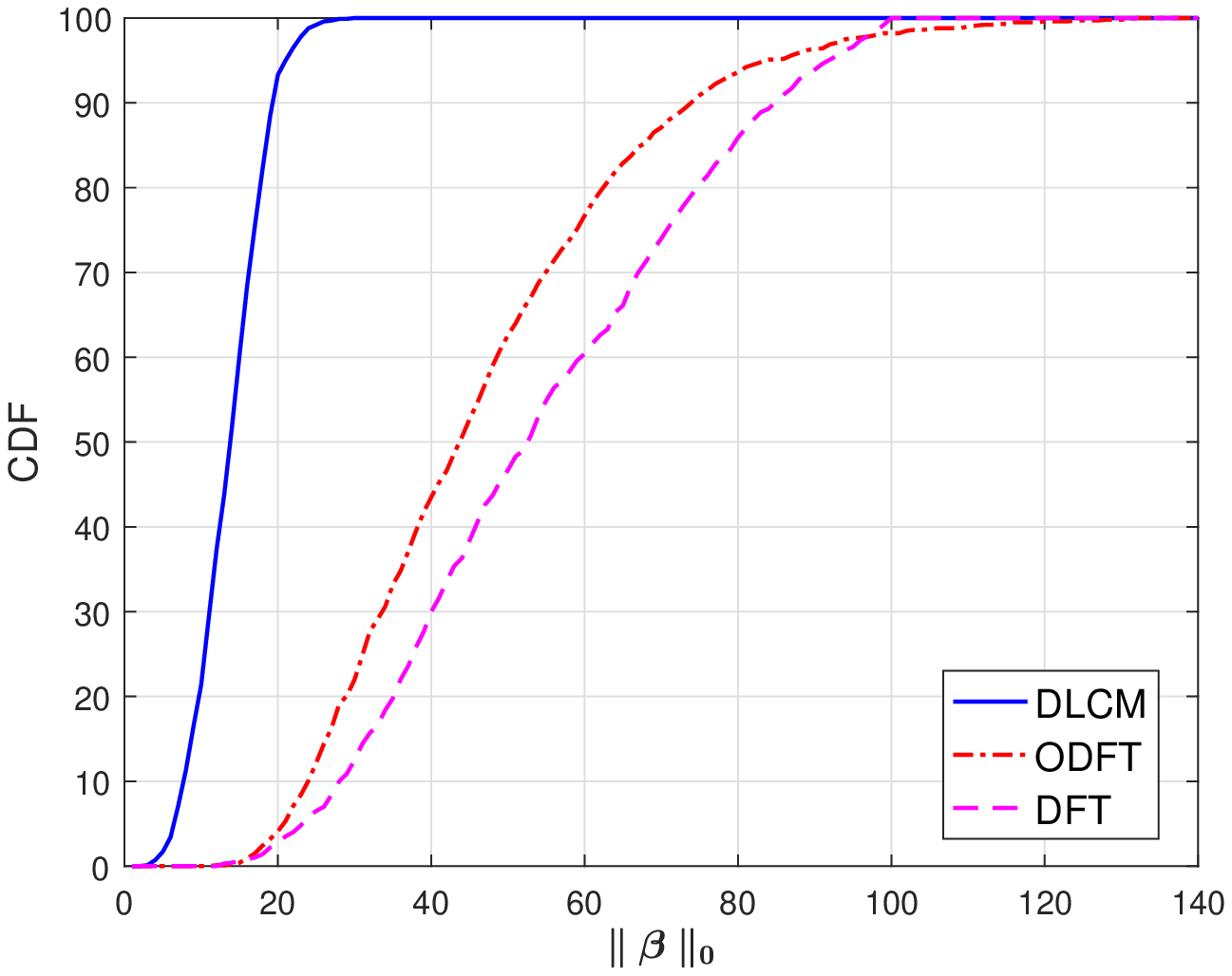}
}
\end{subfigure}
\caption{Cumulative distribution function of \(\|\bm \beta\|_0\). (a) Perfectly calibrated antenna array. (b) Antenna array with uncertainties. \(N=100\), \(\eta=0.1\).}
\label{fig: l0l1}
\end{figure*}

The motivation of using DLCM is to find a dictionary which can (a) more efficiently represent the channel response, i.e., the sparse coefficient has fewer number of nonzero entries; and (b) more robustly represent the channel response, i.e., adapting to any model mismatches like antenna uncertainties.
We generate \(1000\) channel responses \(\bm h_i\), normalize them to have unit norm and calculate the sparse coefficient using \(\ell_1\) framework
\footnote{The sparse representation can be obtained by any sparse recovery algorithm, where different algorithm may lead to different sparse coefficient. We use \(\ell_1\) framework here to be consistent with the algorithm applied in dictionary learning and channel estimation.}
:
\begin{equation}\label{equ: channel representation l1}
\hat{\bm \beta}_i =\arg\min_{\bm \beta_i} \|\bm \beta_i\|_1 \text{ subject to } \|\bm h_i - \bm D\bm \beta_i\|_2\le \eta\\
\end{equation}
where \(\eta\) is set to be $0.1$, so the tolerance of model mismatch is \(10\%\). We then compute \(\|\hat{\bm \beta}_i\|_0\) and plot its cumulative distribution function (CDF) using \(1000\) channel responses. Fig.\ref{fig: l0l1} (a) shows that for perfectly calibrated antenna array, the learned dictionary can represent channel responses using fewer number of nonzero entries. 
For example, \(90\%\) of channel responses can be represented using about \(20\) columns from the learned dictionary, while it requires about 27 or 100 columns if using overcomplete DFT matrix or square DFT matrix. 
In Fig.\ref{fig: l0l1} (b), we test antenna array with uncertainties. Both predefined sparsifying matrices are no longer able to sparsely represent the channel, while the learned dictionary achieve efficient sparse representation similar to Fig.\ref{fig: l0l1} (a).
The results indicate that for a perfectly calibrated antenna array, the suggested overcomplete DFT matrix is a reasonably good sparsifying matrix with only a little inferior than the learned dictionary.
However, when antenna array has uncertainties, predefined matrices degrades considerably due to the huge structure mismatches. In contrast, the learned dictionary leads to efficient and robust representation in both situations, since it is learned from the data without any structure constraint.


\subsection{Downlink Channel Estimation}
\begin{figure*}[!t]
\centering
\begin{subfigure}[]{
\includegraphics[width=0.47\textwidth]{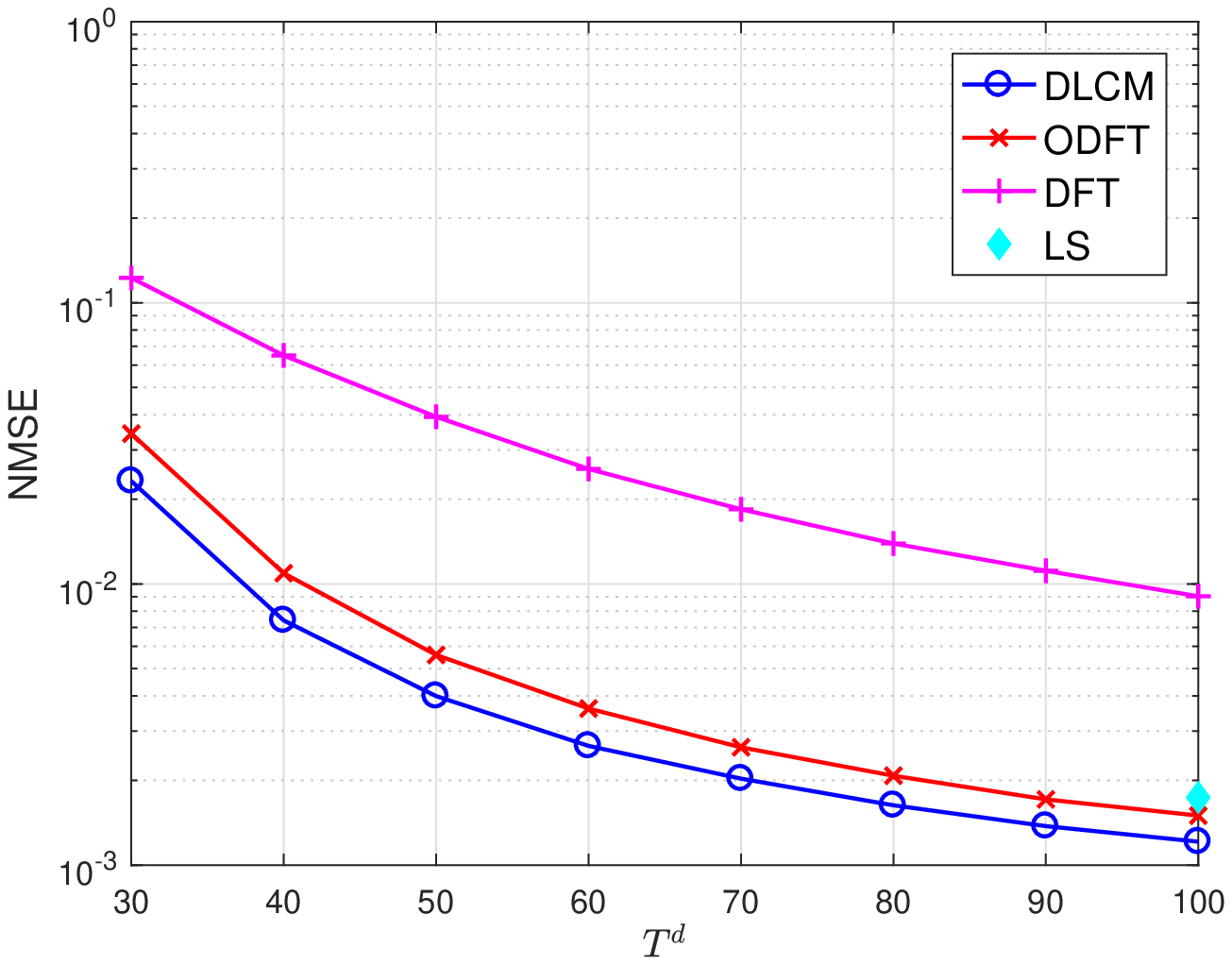}
}
\end{subfigure}
\begin{subfigure}[]{
\includegraphics[width=0.47\textwidth]{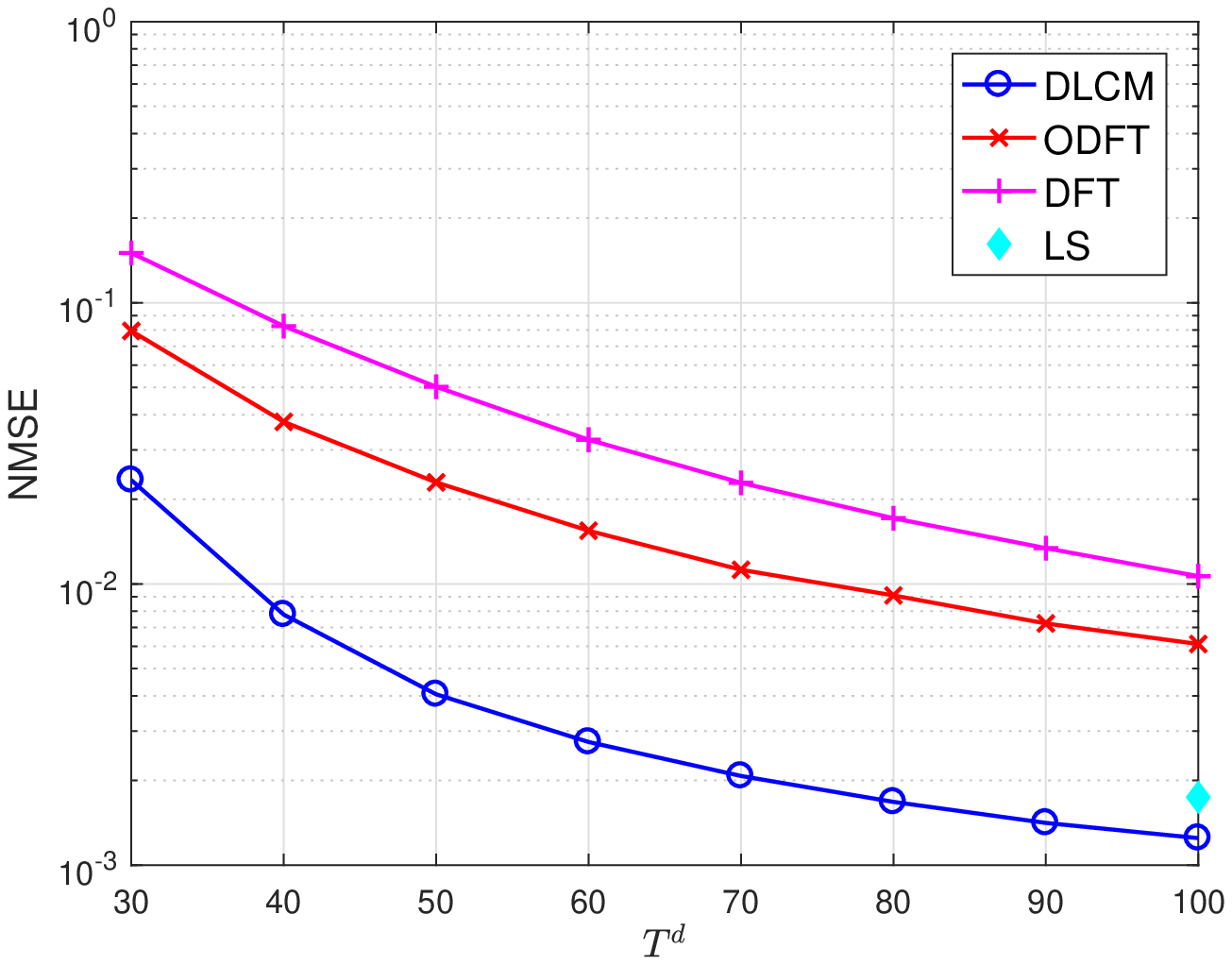}
}
\end{subfigure}
\caption{Normalized mean square error (NMSE) comparison of different sparsifying matrices for CS based downlink channel estimation with ULA. (a) Perfectly calibrated antenna array. (b) Antenna array with uncertainties. SNR = 20dB.}
\label{fig: dl}
\end{figure*}

To evaluate how the channel representation affects the channel estimation, we compare the performance of compressive sensing based downlink channel estimation when different sparsifying matrices are applied.
The training pilots in \(\bm A\) is generated as i.i.d. \(\mathcal{CN}(0, \rho^d/N)\), so \(\mathbb{E}\|\bm A\|_F^2 = \rho^d T^d\). The normalized mean square error (NMSE) is used as the performance metric and defined as \(\text{NMSE}=\mathbb{E}\{\|\bm h-\hat{\bm h}\|_2^2/\|\bm h\|_2^2\}\).
We first consider the 2D channel model with ULA.
Fig.\ref{fig: dl} (a) plots the NMSE performance with respect to the number of downlink pilot symbols \(T^d\), when a perfectly calibrated antenna array is applied. We also include the LS channel estimation when \(T^d=100\) for comparison. 
To achieve the same NMSE, both DLCM and ODFT requires much less training pilots compared to DFT, and the DLCM saves more than ODFT.
Antenna array with uncertainties is tested in Fig.\ref{fig: dl} (b). It shows that the performance of ODFT degrades considerably, while the DLCM achieves almost the same accuracy as in Fig.\ref{fig: dl} (a). So when the antenna array is not perfectly calibrated, only the learned dictionary can achieve great savings on downlink training overhead.

\begin{figure*}[!t]
\centering
\begin{subfigure}[]{
\includegraphics[width=0.47\textwidth]{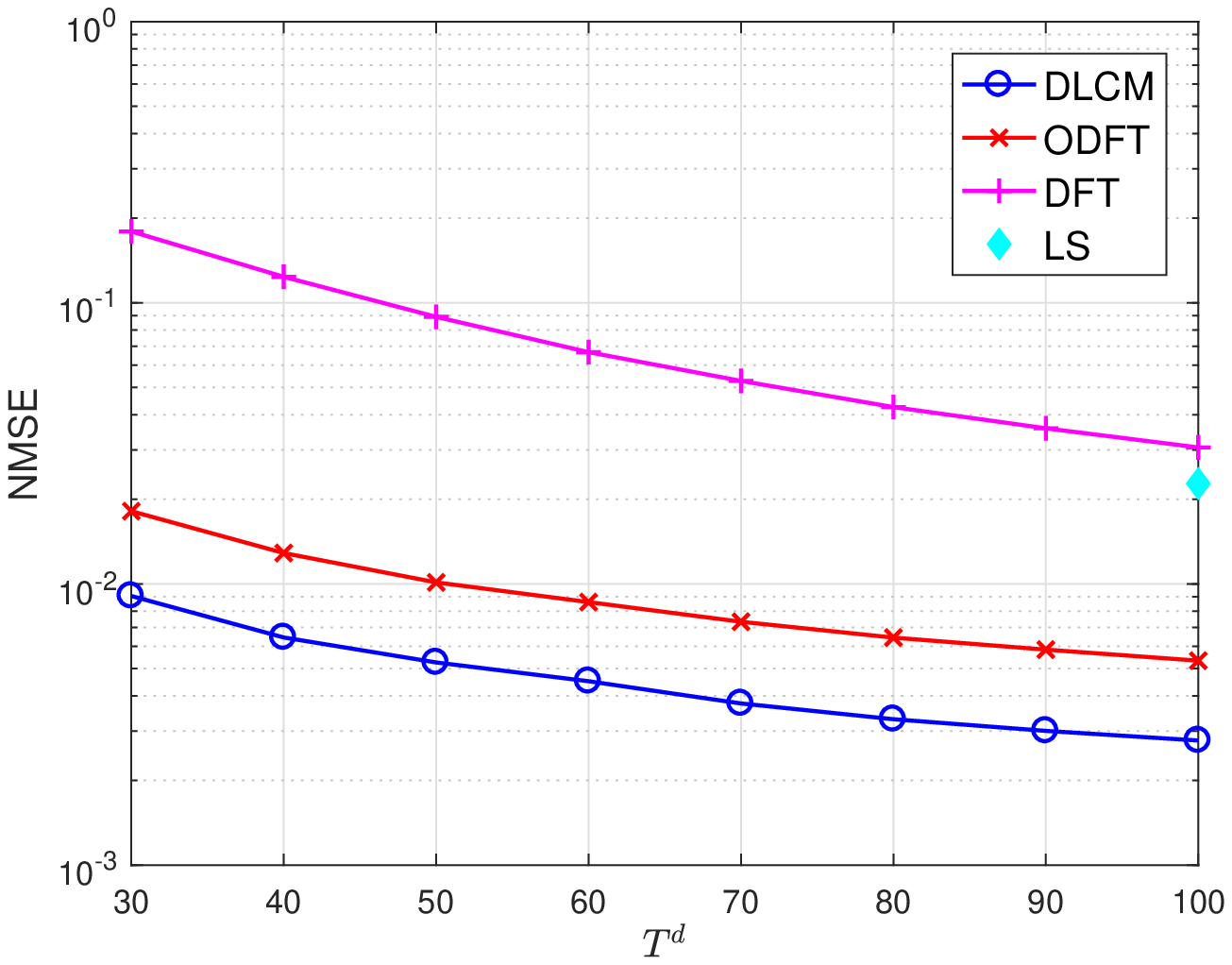}
}
\end{subfigure}
\begin{subfigure}[]{
\includegraphics[width=0.47\textwidth]{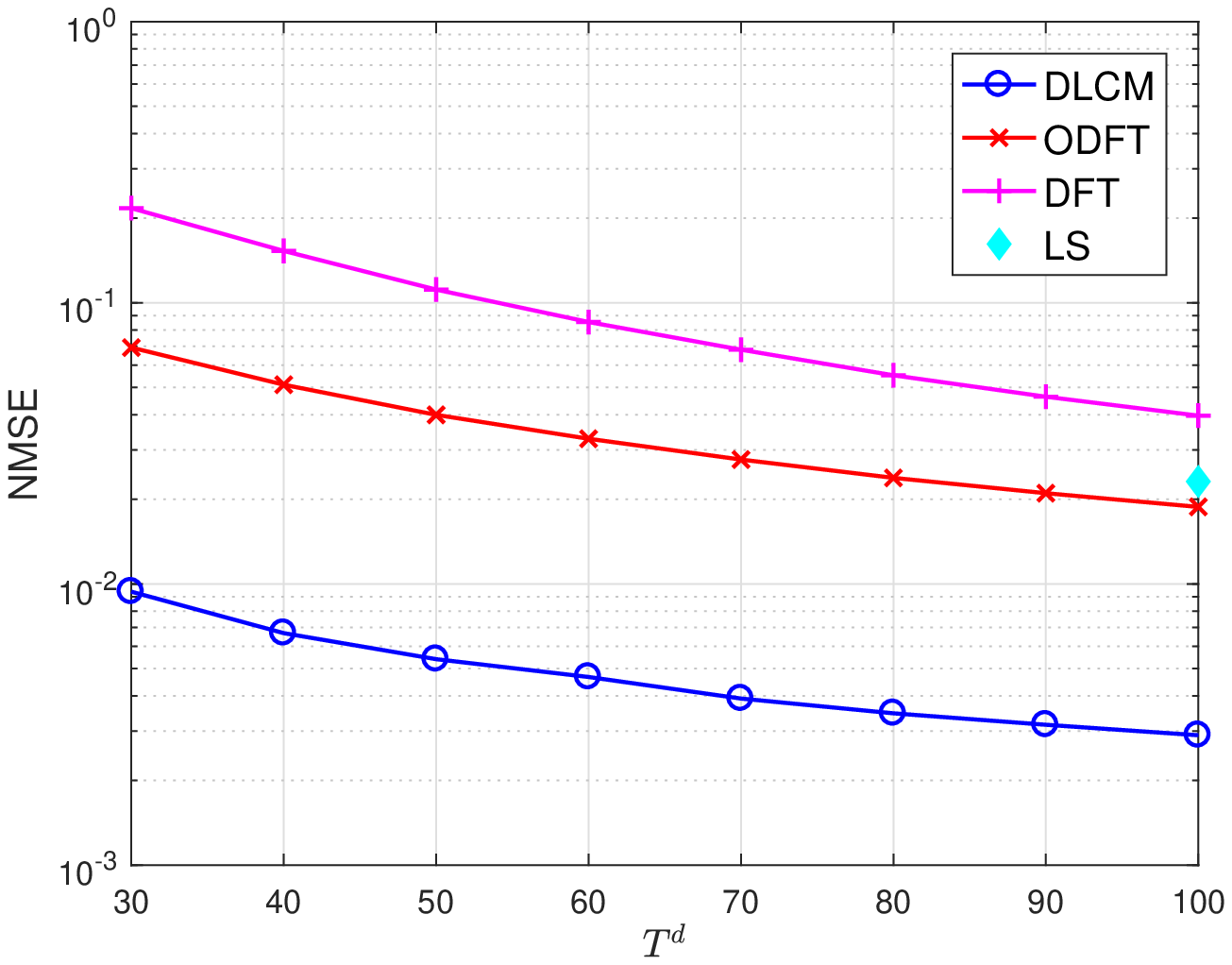}
}
\end{subfigure}
\begin{subfigure}[]{
\includegraphics[width=0.47\textwidth]{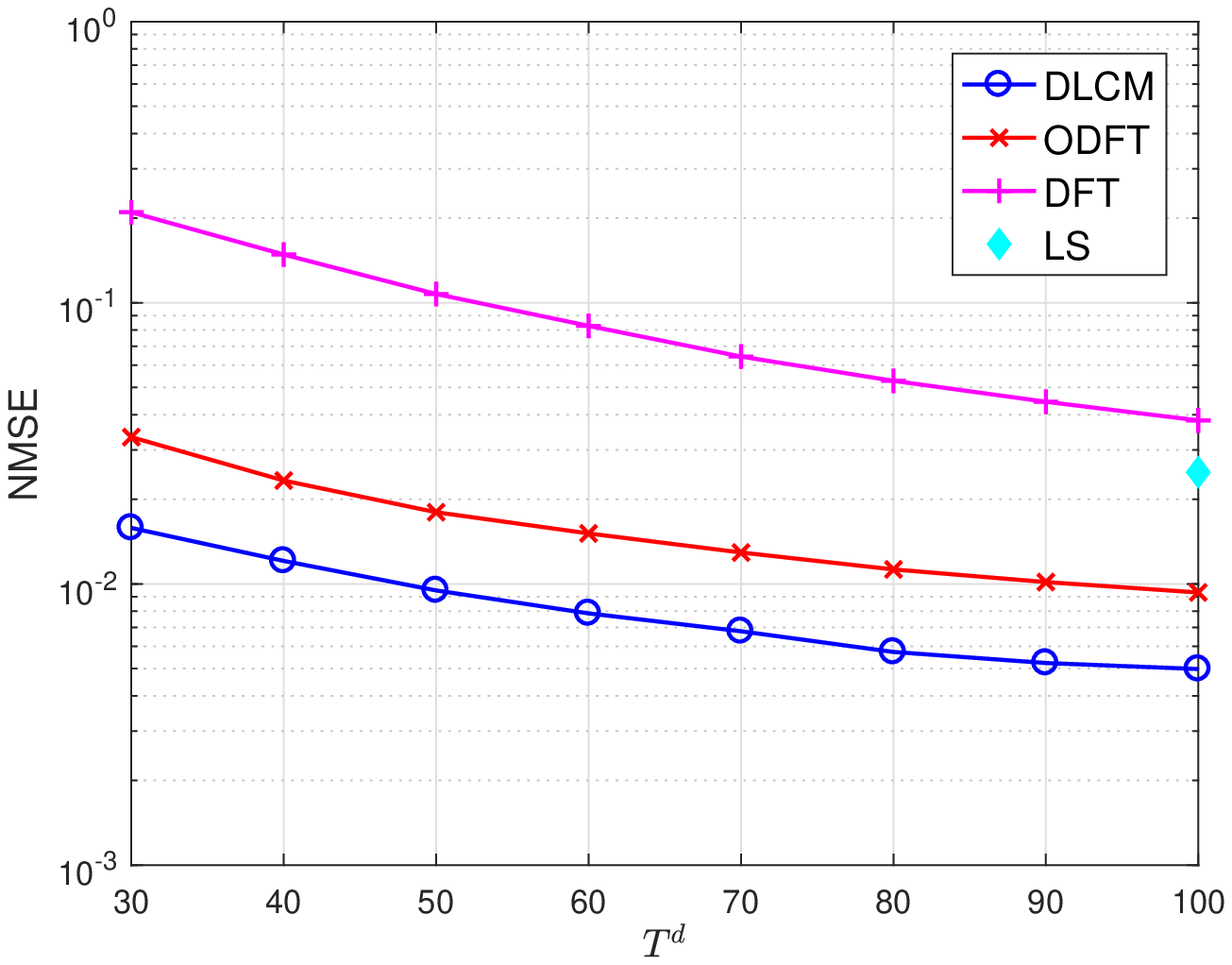}
}
\end{subfigure}
\begin{subfigure}[]{
\includegraphics[width=0.47\textwidth]{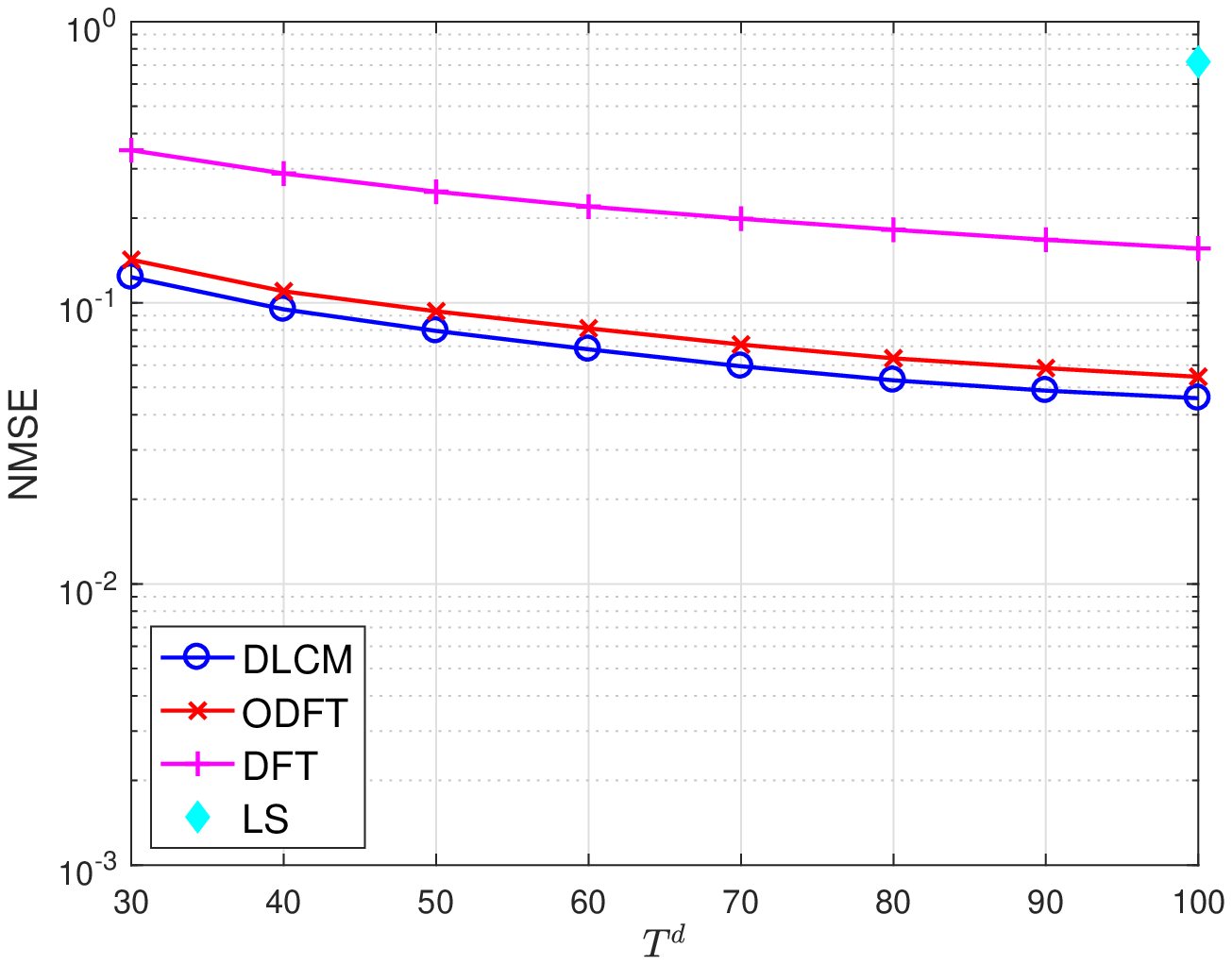}
}
\end{subfigure}
\caption{Normalized mean square error (NMSE) comparison of different sparsifying matrices for CS based downlink channel estimation with URA in mmWave scenario. (a) Perfectly calibrated antenna array. \(\text{BS: } 10\times 10, \text{UE: } 3\times 3, \text{SNR = 20 dB}\).  (b) Antenna array with uncertainties. \(\text{BS: } 10\times 10, \text{UE: } 3\times 3, \text{SNR = 20 dB}\). (c) Perfectly calibrated antenna array. \(\text{BS: } 10\times 10, \text{UE: } 1\times 1, \text{SNR = 20 dB}\). (d) Perfectly calibrated antenna array.  \(\text{BS: } 10\times 10, \text{UE: } 3\times 3, \text{SNR = 5 dB}\).}
\label{fig: mmWave}
\end{figure*}

In Fig.\ref{fig: mmWave}, we test mmWave 3D channel model with a \(10\times 10\) URA at the BS and a \(3\times 3\) URA at the UE. The MMV channel estimation in (\ref{equ: MMV minimization}) is applied with respect to multiple antennas at the UE. 
Performance comparison with perfectly calibrated antenna array and antenna array with uncertainties are shown in Fig.\ref{fig: mmWave} (a) and Fig.\ref{fig: mmWave} (b). Similar to the 2D channel model with ULA, the DLCM achieves better performance than ODFT and DFT. The result demonstrates the applicability of the proposed DLCM framework to different antenna geometry and frequency band. Fig.\ref{fig: mmWave} (c) plots the performance with only a single antenna at the UE. Compared to Fig.\ref{fig: mmWave} (a) where the UE has \(9\) antennas, the performance becomes worse. This shows the benefit of having multiple antennas at UE and utilizing the proposed MMV formulation (\ref{equ: MMV minimization}) to estimate the channel. To study the channel estimation in low SNR range, Fig.\ref{fig: mmWave} (d) depicts the performance when \(\text{SNR} = 5 \text{dB}\). Compared to Fig.\ref{fig: mmWave} (a) where \(\text{SNR} = 20 \text{dB}\), the performances of all sparsifying matrices are worse, and the differences among them become small. The reason is that when the noise is large, the accuracy of the channel estimation is limited mostly by the noise, so the model mismatch error from applying different sparsifying matrices has only small influence on the performance. 
In a practical system, performance of channel estimation depends on many factors such as the noise level, model mismatch error, and the number of antennas. More studies are required to show under what condition the DLCM can achieve the greatest improvement.
\begin{figure*}[!t]
\centering
\begin{subfigure}[]{
\includegraphics[width=0.47\textwidth]{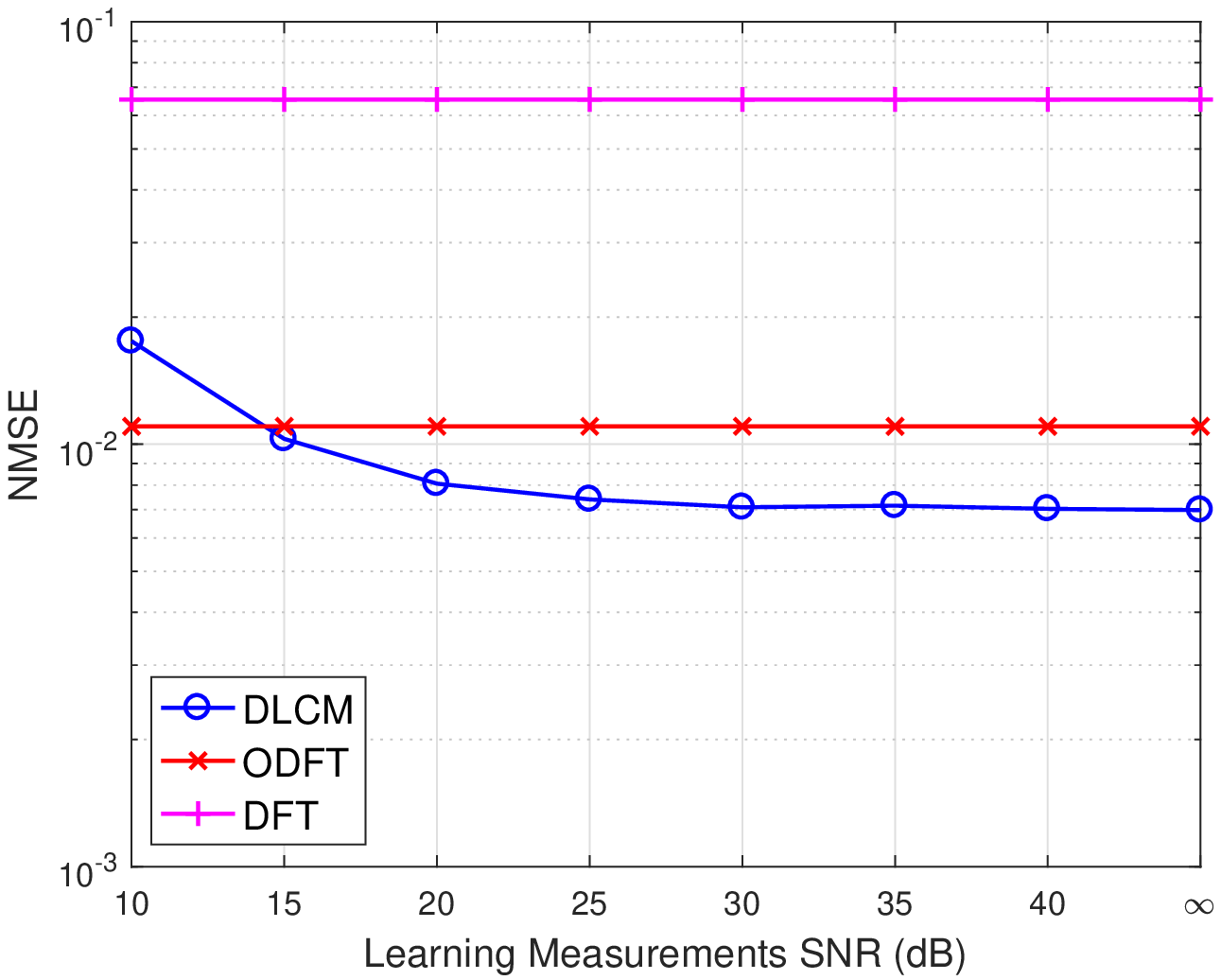}
}
\end{subfigure}
\begin{subfigure}[]{
\includegraphics[width=0.47\textwidth]{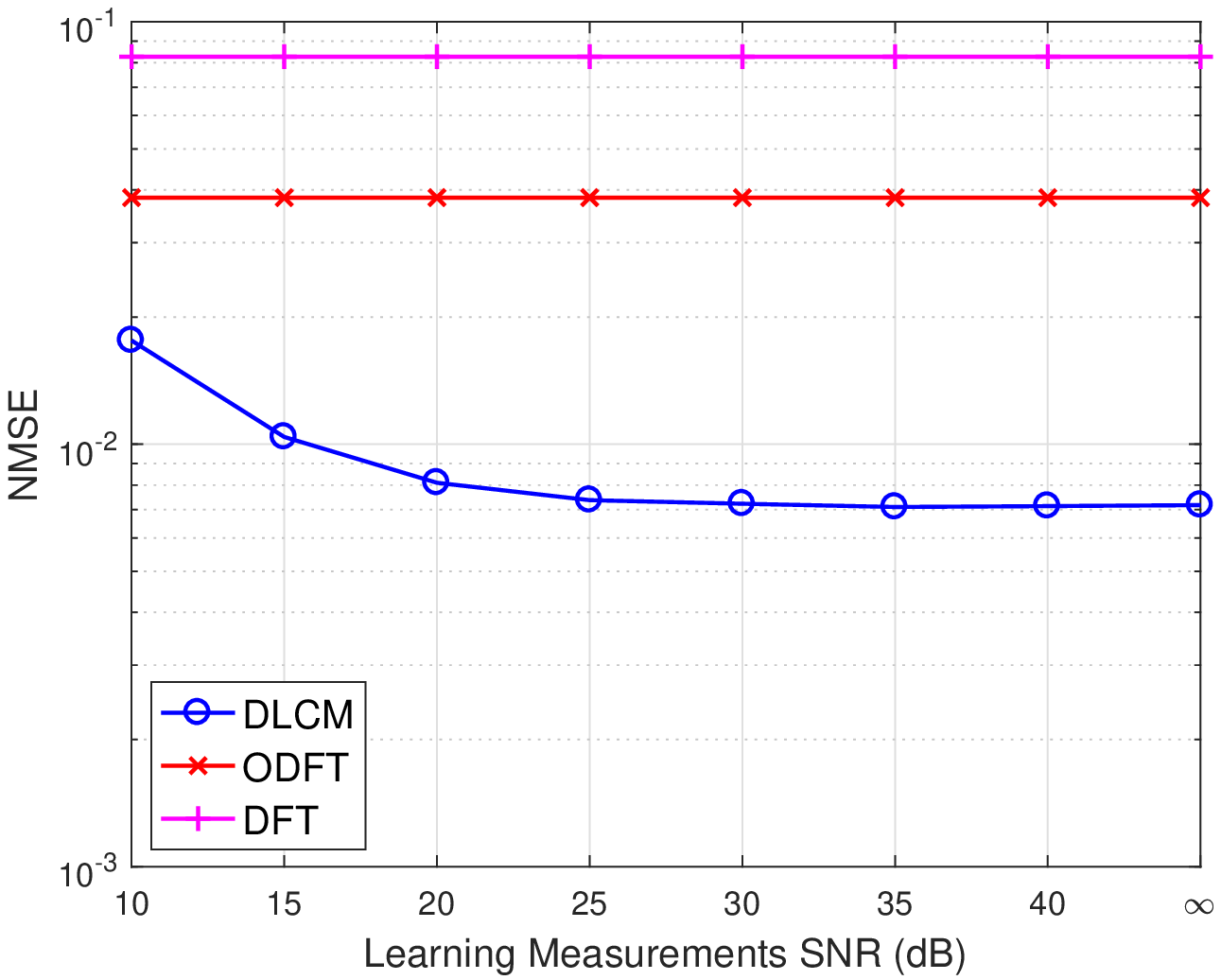}
}
\end{subfigure}
\caption{Normalized mean square error (NMSE) versus learning measurements SNR for CS based downlink channel estimation with ULA. (a) Perfectly calibrated antenna array. (b) Antenna array with uncertainties. \(T^d=40\), SNR = 20dB.}
\label{fig: comp_dic}
\end{figure*}

Next we investigate the performance of DLCM when the dictionary is learned from channel responses corrupted by the noise, since in practice channel measurements can have some estimation error in them. We consider the 2D channel model using ULA, and add noise \(\tilde{\bm n}_i\) to the true channel response \(\bm h_i\), where \(\tilde{\bm n}_i\sim\mathcal{CN}(\mathbf{0},\sigma^2_i\mathbf{I})\) and \(\sigma_i\) is chosen according to the ``learning measurements SNR'' which is defined as \(\|\bm h_i\|_2^2/{\mathbb{E}\|\tilde{\bm n}_i\|_2^2}=\|\bm h_i\|_2^2/(N\sigma_i^2)\). Then the corrupted channel response is obtained as \(\tilde{\bm h}_i = \bm h_i+ \tilde{\bm n}_i\), which is used as the training channel samples for the dictionary learning. 
Fig. \ref{fig: comp_dic} compares the NMSE of downlink channel estimation with respect to the learning measurements SNR, where for each learning measurements SNR a different dictionary is learned. 
We also include the performance when the true channel response \(\bm h_i\) is used as the training samples, and denoted it as \(\infty\) learning measurements SNR. When the learning measurements SNR is low, the performance of DLCM degrades since the dictionary learning process can not accurately capture the channel structure from too noisy channel measurements. As the learning measurements SNR increases, the performance of DLCM becomes better and approaches the performance of learning from noiseless measurements. Notice that when the antenna array is not perfectly calibrated, as shown in Fig. \ref{fig: comp_dic} (b), DLCM can obtain better performance than predefined sparsifying matrices even with dictionary learned from very noisy measurements. Since the dictionary learning is performed at the cell deployment stage, a high learning measurements SNR can be achieved by using more training pilots, higher transmitted power, and more sophisticated equipment. As a result, the learned dictionary is expected to efficiently represent the sparse channel and lead to a good channel estimation performance using DLCM.

\subsection{Uplink Channel Estimation}
\begin{figure*}[!t]
\centering
\begin{subfigure}[]{
\includegraphics[width=0.47\textwidth]{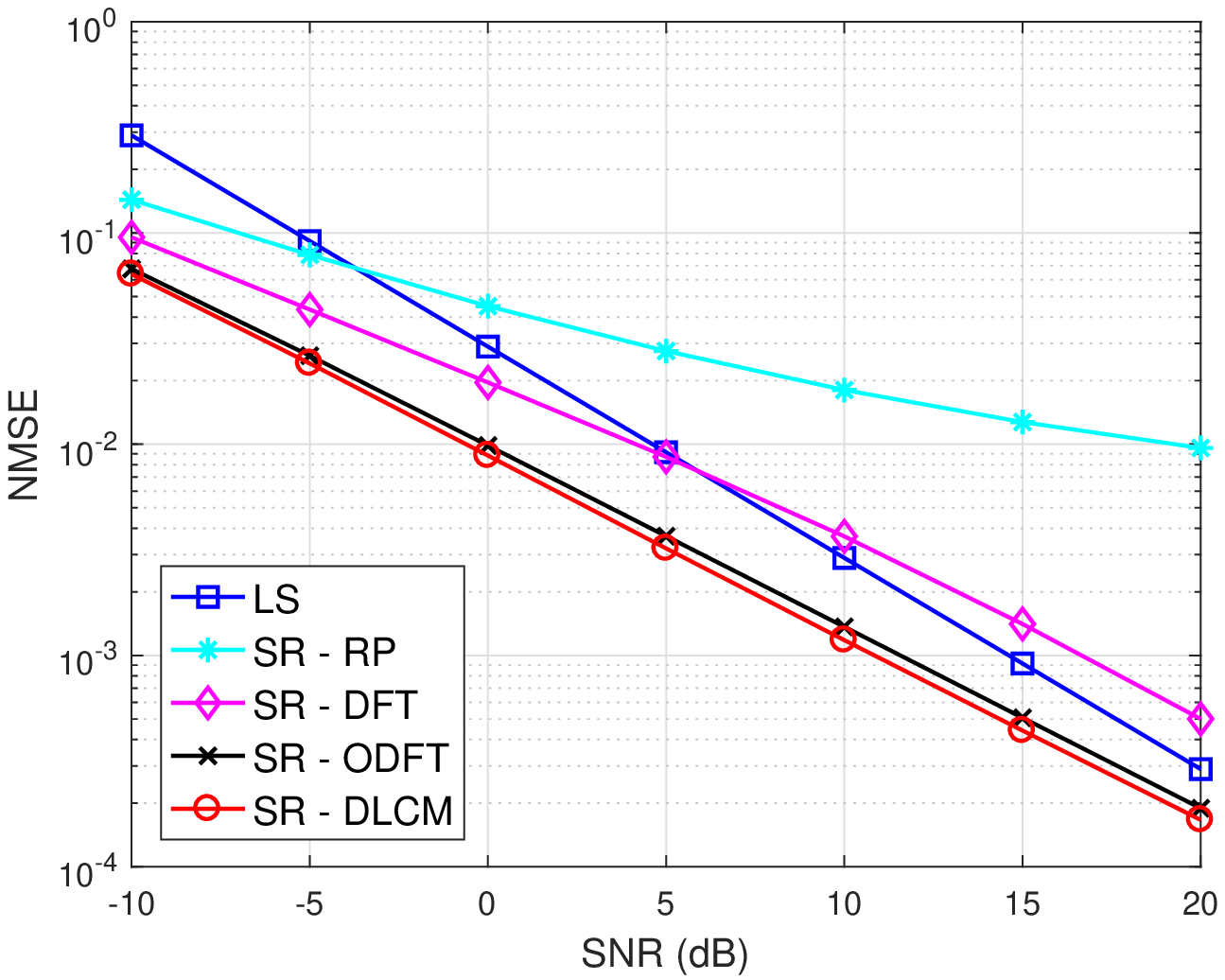}
}
\end{subfigure}
\begin{subfigure}[]{
\includegraphics[width=0.47\textwidth]{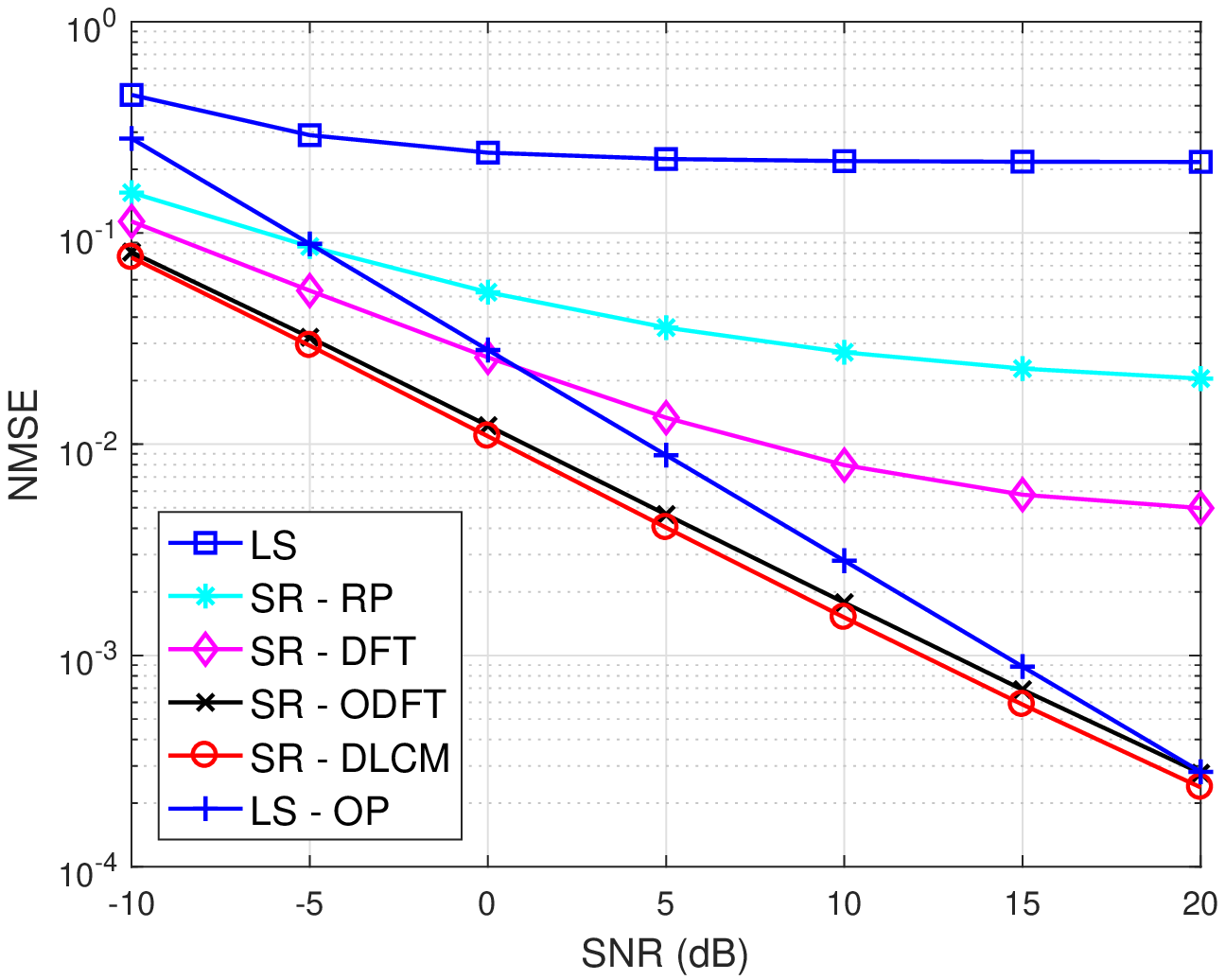}
}
\end{subfigure}
\begin{subfigure}[]{
\includegraphics[width=0.47\textwidth]{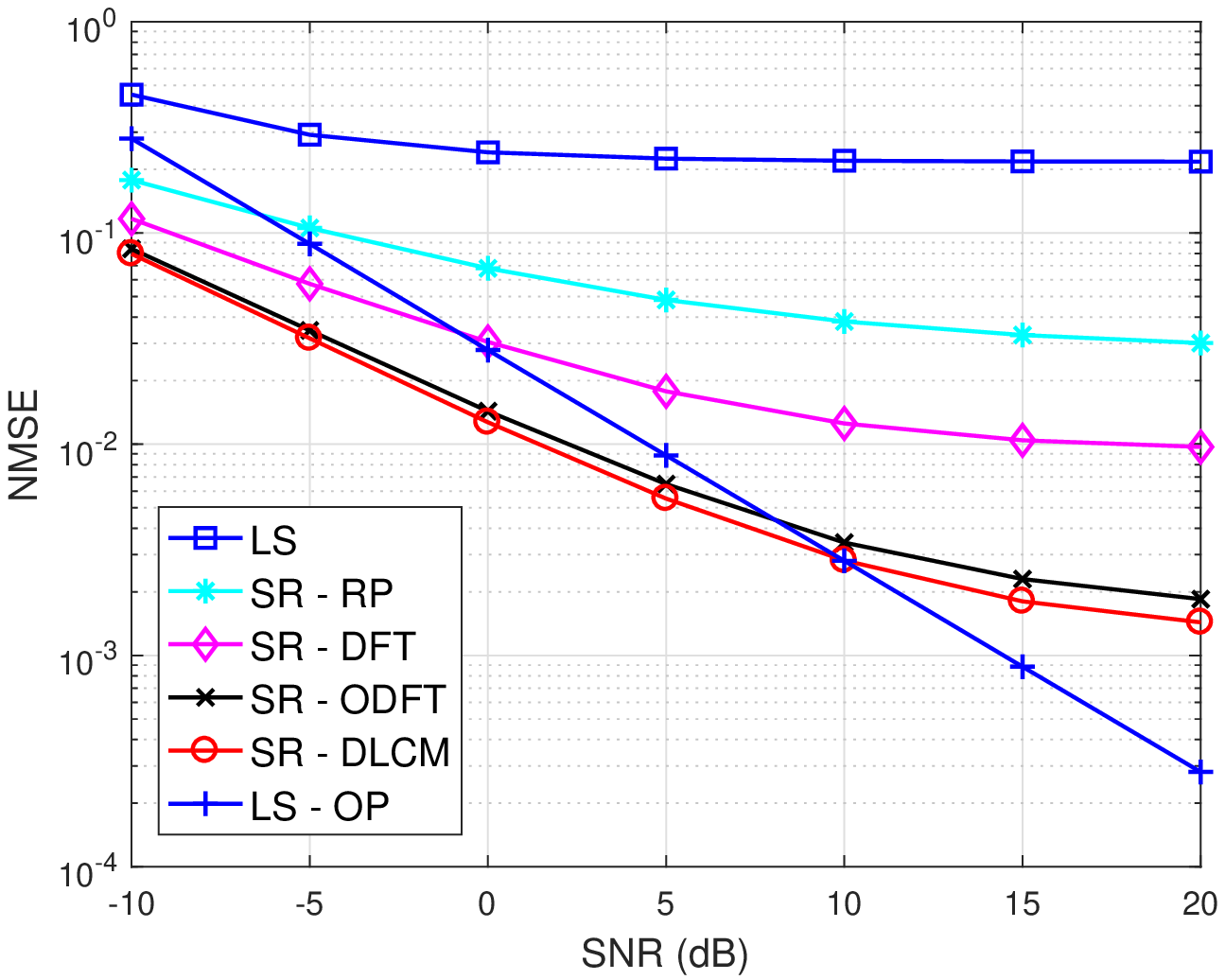}
}
\end{subfigure}
\begin{subfigure}[]{
\includegraphics[width=0.47\textwidth]{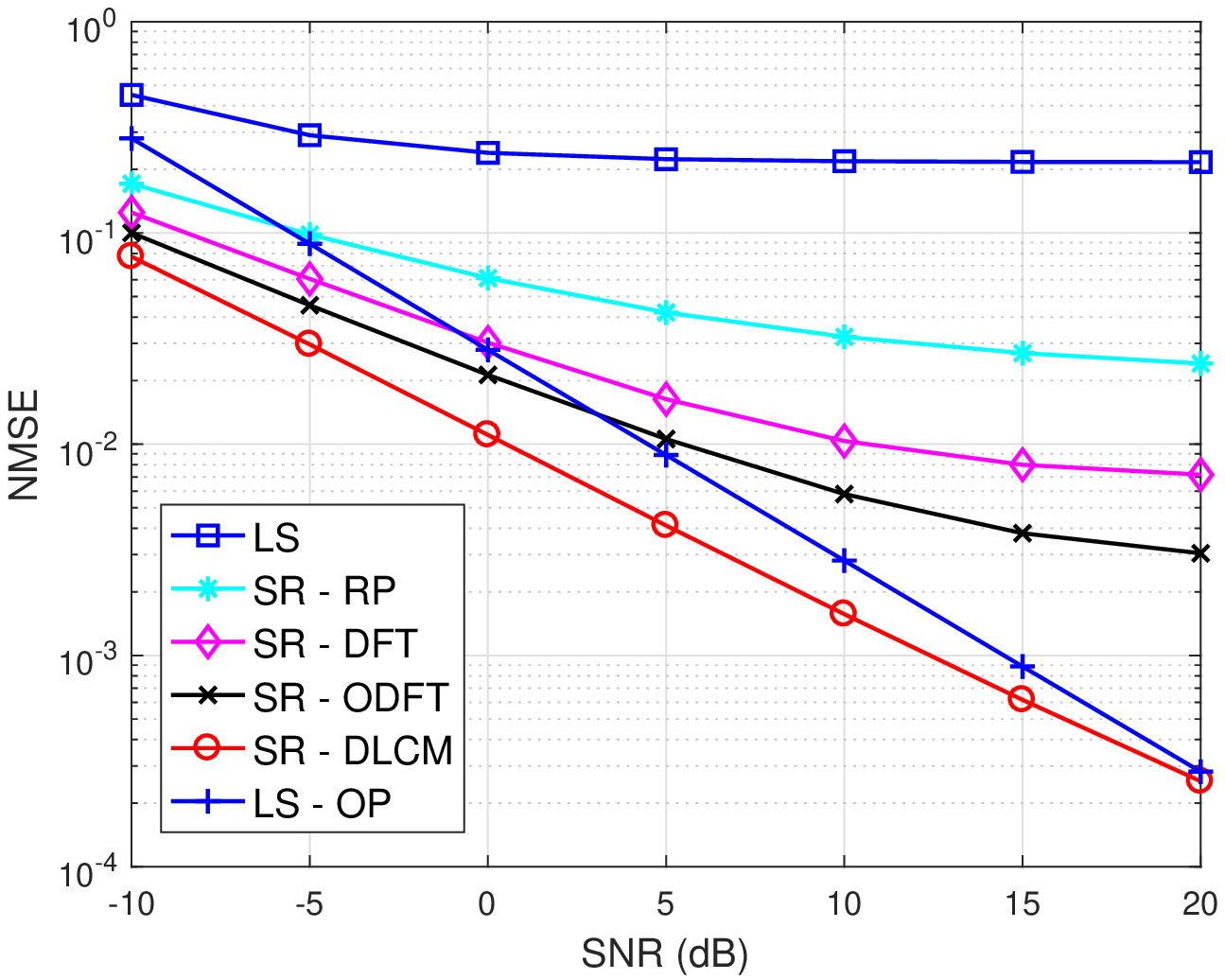}
}
\end{subfigure}
\caption{Normalized mean square error (NMSE) comparison of different sparsifying matrices, pilots design and user scheduling schemes for SR based uplink channel estimation with ULA. (a) Orthogonal users, perfectly calibrated antenna array. \(K=6, T^u=6\). (b) Orthogonal users, perfectly calibrated antenna array. \(K=6, T^u=5\). (c) Non-orthogonal users, perfectly calibrated antenna array. \(K=6, T^u=5\). (d) Orthogonal users, antenna array with uncertainties. \(K=6, T^u=5\).}
\label{fig: ul}
\end{figure*}

We now evaluate performance of the sparse recovery (SR) based uplink channel estimation using 2D channel model. Assume there are \(K=6\) users. 
For simplicity of illustration, we assume the same \(\rho^u_k\) for all users and plot the average NMSE versus the SNR
\footnote{We also performed simulations where different \(\rho^u_k\) is assigned to each user, and evaluate each user's NMSE performance separately. Similar conclusions can be made as the same \(\rho^u_k\) scenario.}.
In Fig.\ref{fig: ul} (a), the number of uplink pilots \(T^u=K=6\), 
so orthogonal pilots are used for both LS and SR with different sparsifying matrices. We also compare SR with random pilots (SR-RP), which is possibly nonorthogonal. To encourage non-overlapping or limited overlapping sparse support, the locations of users are generated to be far away from each other, with each user's LOS AOA constrained in a distinct \(\pi/K\) range. 
For the whole SNR range, SR-DLCM and SR-ODFT are better than LS, while SR-DFT is worse than LS at high SNR. 
SR-RP can not achieve good performance, since random pilots can not lead to small \(\mu\{\bm E\}\). 
In Fig.\ref{fig: ul} (b), \(T^u=5\) is tested. The performance of LS degrades a lot, since \(T^u<K\) and the problem is underdetermined for LS estimation. On the other hand, SR with pilots design suggested in Section \ref{subsection: uplink channel estimation} has only little degradation. LS using orthogonal pilots (LS-OP, where \(T^u=6\)) is also provided for comparison. Notably, SR-DLCM and SR-ODFT with \(T^u =5\) can achieve even better performance than LS with \(T^u =6\), showing the great benefit of using SR for uplink channel estimation. 
In Fig.\ref{fig: ul} (c), users' locations are randomly and uniformly generated, so their supports can possibly be overlapped a lot. In this case the SR based channel estimation degrades severely at high SNR, indicating the importance of minimizing \(\mu\{\bm E_{\Lambda}\}\) in order to achieve good sparse recovery performance in (\ref{equ: uplink sparse reocvery}). 
In Fig.\ref{fig: ul} (d), the antenna array with uncertainties is used to show the robustness of the learned dictionary. 

Fig.\ref{fig: ul} shows the benefits of utilizing sparse property to perform the uplink channel estimation, and the essential requirements are (a) sparsifying matrix which can lead to efficient and robust sparse representation; (b) pilots design scheme which minimizes \(\mu\{\bm E\}\); and (c) user scheduling scheme which decreases \(\mu\{\bm E_{\Lambda}\}\). Even with \(T^u<K\), SR based channel estimation can still achieve good performance. The experiments in Fig.\ref{fig: ul} consider the single cell scenario, but can be easily extended to multi-cell scenario for pilot decontamination. For example, consider $6$ cells and each of them has $6$ users. Assume the total uplink training duration constraint \(T^u=30\). If LS channel estimation is applied, then each cell requires at least $6$ training duration, so for all $6$ cells their pilots can not be orthogonal to each other and pilot contamination occurs. However, by using SR based channel estimation, each cell requires only $5$ training duration to achieve the similar (even better) performance than LS. Training duration of $30$ is enough for $6$ cells to have orthogonal pilots, so there is no pilot contamination anymore.

\subsection{Channel Estimation Using Jointly Learned Dictionary}
\begin{figure*}[!t]
\centering
\begin{subfigure}[]{
\includegraphics[width=0.47\textwidth]{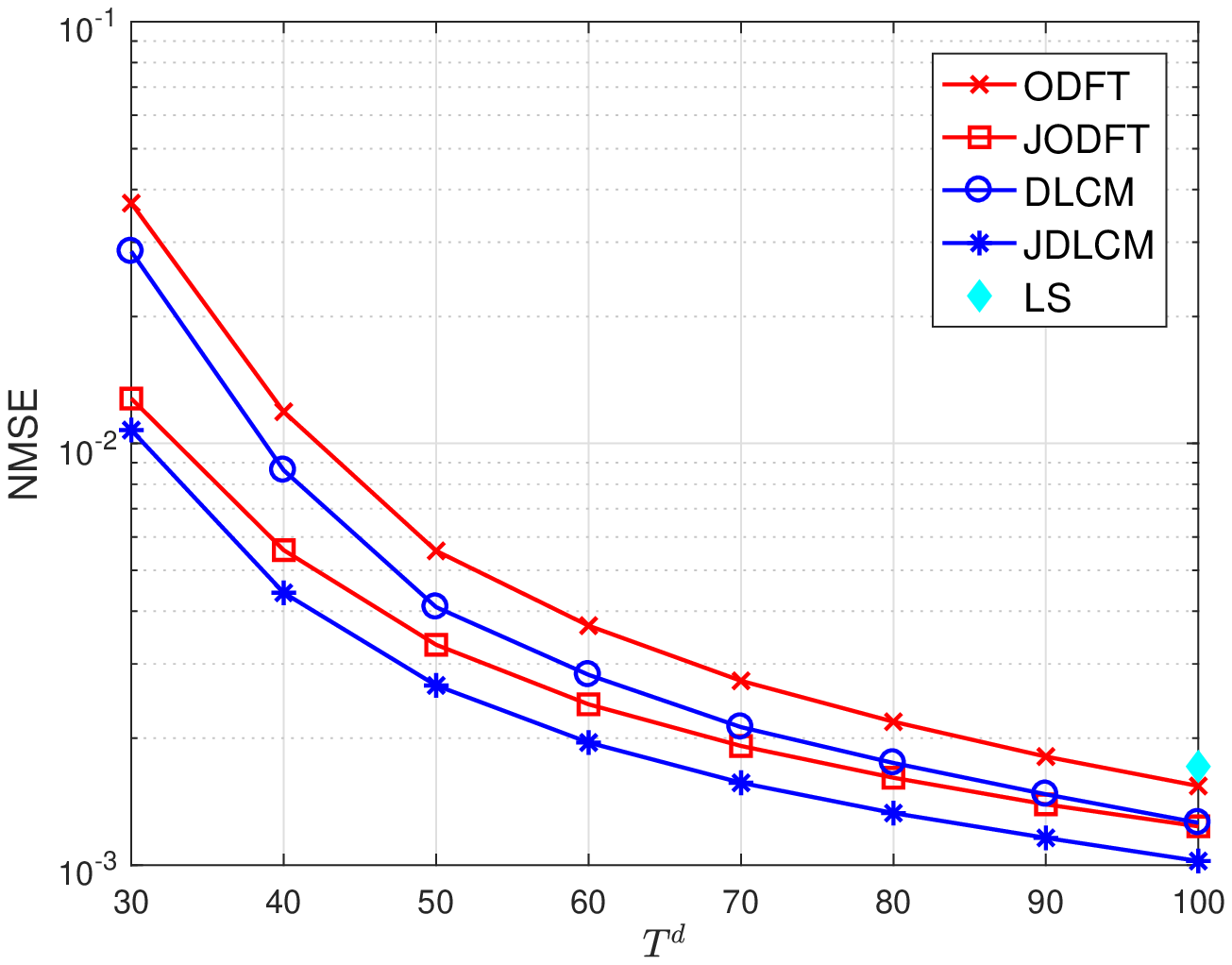}
}
\end{subfigure}
\begin{subfigure}[]{
\includegraphics[width=0.47\textwidth]{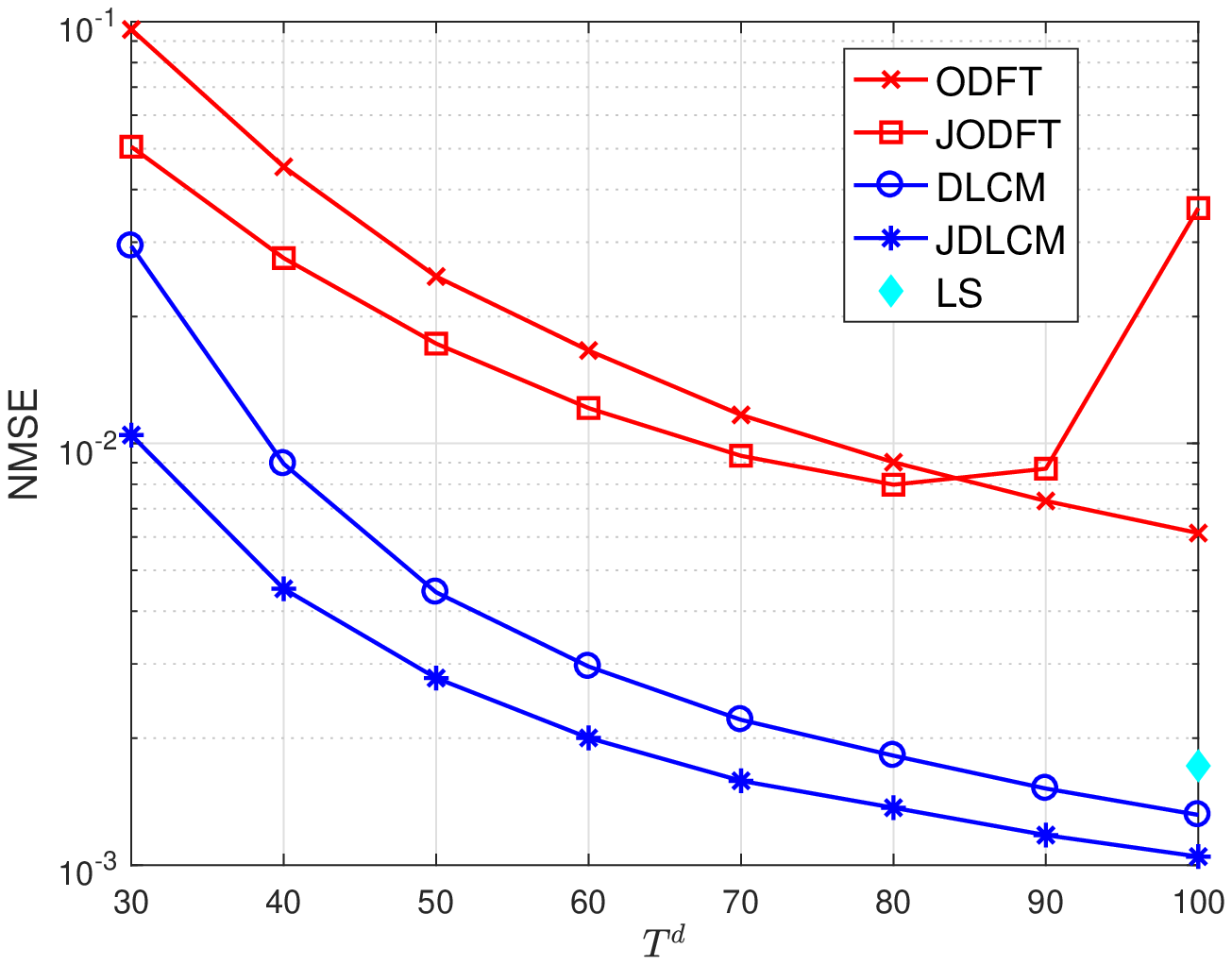}
}
\end{subfigure}
\caption{Normalized mean square error (NMSE) comparison of different sparsifying matrices for joint channel estimation with ULA. (a) Perfectly calibrated antenna array. (b) Antenna array with uncertainties. SNR = 20dB.}
\label{fig: judl}
\end{figure*}

For joint channel estimation, assume the uplink frequency is 1920 MHz and downlink frequency is 2110 MHz. The antenna spacing \(d=\frac{c}{2f_0}\) where \(c\) denotes the light speed and \(f_0=2010\) MHz. 
During channel estimation, we set \(T^u =2\), and \(\rho^u=\rho^d\). 
Fig. \ref{fig: judl} (a) compares downlink joint and independent channel estimation performance. Besides the jointly learned dictionary (JDLCM), we consider the joint overcomplete DFT matrix (JODFT) obtained by setting \(\psi^u=d\text{sin}(\theta)/\lambda^u\), \(\psi^d=d\text{sin}(\theta)/\lambda^d\) and \(\text{sin}(\theta)\in[-1,-1+\frac{2}{M},\ldots,1-\frac{2}{M}]\) as in (\ref{equ: overcomplete DFT}). Smaller NMSE can be obtained by the joint channel estimation compared to their independent counterpart. Such improvement is most obvious when \(T^d\) is small, since the additional measurements from the uplink training help a lot. Fig. \ref{fig: judl} (b) shows the robustness of JDLCM when there exist antenna uncertainties. The JODFT is no longer applicable in this case since the structure is incorrect, and for the large \(T^d\) it becomes even worse than the ODFT.
With the help of small number of uplink training (\(T^u=2\) in the experiment), one can further improve the performance of downlink channel estimation therefore reduce the downlink training overhead. 
The simulation is conducted in the microwave scenario. More investigation, especially real experimental measurements, are needed to support the uplink/downlink angular reciprocity in a mmWave scenario. 

\section{Conclusion}\label{section: conclusion}
In this paper, we develop a dictionary learning based channel model which learns a cell specific dictionary from comprehensively collected channel measurements from different locations in the cell.
The learned dictionary is able to adapt to the cell characteristics and any antenna array uncertainties, leading to a more efficient and robust channel representation compared to predefined sparsifying matrices.
For both CS based downlink and SR based uplink channel estimation, the learned dictionary can improve the performance and reduce the training overhead.
Motivated by the angular reciprocity between the uplink and downlink channel responses, we further develop a joint dictionary learning based channel model in order to utilize the relatively simpler uplink channel training to help improving the downlink channel estimation.
The results of this paper show that concepts of utilizing sparse property and learning from the data can be useful for future communication systems. 

As future work, several topics are under consideration. To learn the dictionary, extensive channel measurements are needed as the training samples. Besides using conventional drive tests to collect data, minimization of drive tests (MDT), specified in 3GPP release 10 \cite{hapsari2012minimization}, is a promising approach. The main concept is to exploit commercial user equipments, such as their measurement capabilities and geographically spread nature, for collecting radio measurements. Another option is to explore online dictionary learning \cite{mairal2009online}, 
where an initial dictionary is first learned from limited training samples, and then updated as more training samples are obtained.
Online dictionary learning can also be used to deal with the slowly changing cell and antenna characteristics, and adapt to specific user distribution properties in the cell, which is hard to be captured at cell deployment stage.  
Finally, for joint channel estimation, a looser relationship between the supports of the uplink and downlink sparse coefficients may be utilized instead of the strict constraint \(\text{supp}(\bm \beta^u)=\text{supp}(\bm \beta^d)\) to better model the angular reciprocity, for example allowing some mismatch between \(\text{supp}(\bm \beta^u)\) and \(\text{supp}(\bm \beta^d)\) through a Bayesian formulation.

\ifCLASSOPTIONcaptionsoff
  \newpage
\fi

\bibliographystyle{IEEEtran}
\bibliography{DLCMbib}

\end{document}